\documentclass[11pt, a4paper]{article}

\usepackage[utf8]{inputenc}
\usepackage[fleqn]{amsmath}
\usepackage[mathscr]{euscript}
\usepackage{gensymb,amsfonts}
\usepackage{amssymb}
\usepackage{array}
\usepackage[T1]{fontenc}
\usepackage[margin=1in]{geometry}
\usepackage{lmodern}
\usepackage{tabularx}
\usepackage{fancyhdr}
\usepackage{graphicx}
\graphicspath{{images/}}
\usepackage[utf8]{inputenc}
\usepackage[english]{babel}
\usepackage{amsthm}
\usepackage{braket}
\usepackage{tikz}
\usepackage{subfig, float}
\usepackage[utf8]{inputenc}
\usepackage[backend=biber, bibstyle=numeric, url=false, doi=false, sortcites]{biblatex}
\usepackage{authblk,csquotes}
\usetikzlibrary{cd}

\newtheorem{theorem}{Theorem}
\newtheorem{proposition}{Proposition}
\newtheorem{lemma}{Lemma}
\newtheorem{definition}{Definition}
\newtheorem{corollary}{Corollary}

\newcommand{\hlf}{\frac{1}{2}}

\newcommand{\supp}{\text{Supp}}

\newcommand{\mc}{\mathcal}
\newcommand{\spn}{\mathrm{Span}}

\newcommand{\mb}{\mathbb}
\newcommand{\id}{\text{id}}

\newcommand{\im}{\mathrm{Im}\,}

\title{Gauge Structures:\\ From Stabilizer Codes to Continuum Models}
\author{Albert T. Schmitz\footnote{contact: albert.schmitz.colorado.edu}} 
\affil{Department of Physics and Center for Theory of Quantum Matter\\University of Colorado, Boulder, CO 80309, USA}

\begin{document}
\maketitle

\abstract{Stabilizer codes are a powerful method for implementing fault-tolerant quantum memory and in the case of topological codes, they form useful models for topological phases of matter. In this paper, we discuss the theory of stabilizer codes as a discrete version of a {\it linear gauge structure}, a concept we introduce here. A linear gauge structure captures all the familiar ingredients of $U(1)$ gauge theory including a generalization of charge conservation, Maxwell equations and topological terms. Using this connection, we prove some important results for stabilizer codes which can be used to characterize their error-correction properties. However, this perspective does not depend on any particular Hamiltonian or Lagrangian, which as a consequence is agnostic to the source of the gauge redundancy. Based upon the connection to stabilizer codes, we view the source of the gauge redundancy as an ambiguity in the tensor product structure of the Hilbert space. That is, the gauge provides a set of equivalent but arbitrary ways to factorize the Hilbert space. We apply this formalism to the $d=2$ and $3$ toric code as well as the paradigm fracton models, X-Cube and Haah's cubic code. From this perspective, we are able to map all these models to some continuum version which captures all or some of the fractonic and topological features of their parent discrete models.
}
      
\section{Introduction}\label{sec:intro}
Gauge theory has been a successful means of describing nature. Though first thought of as a symmetry, the more modern view is that it reflects a redundancy in the description of the physical system. However, there must be more to the theory than just our labels have more ``letters'' than necessary. Gauge theory captures how we are forced to make a necessary but arbitrary choice. Furthermore, we expect these redundancies have real observable consequences. 

Continuum gauge theory is most naturally couched in the language of differential geometry and Lie algebras as is the case for Yang-Mills theories. However, this clearly doesn't translate to discrete and finite examples. It is desirable to frame gauge theory in a categorical way such that there is a common language for anything that might be called a gauge theory. To advance toward this goal, we look to find a common thread between the continuum and discrete theories. For the discrete case, we focus on topological stabilizer codes which are best exemplified by Kitaev's toric code \cite{Kitaev2003}. Stabilizer codes as introduced in \cite{Gottesman1997} are one form of fault-tolerant quantum memory, but can also act as useful, exactly-solvable models for topological phases of matter. An interesting subset of these codes are the so-called {\it fracton} stabilizer codes as first  introduced in Refs.~\cite{Chamon2005,Nussinov2009,Nussinov2009a}, made prominent in Refs.~\cite{Bravyi2011,Haah2011,Yoshida2013,Vijay2015,Vijay2016} and elaborated upon in Refs.~\cite{ Williamson2016, Pretko2017a, Pretko2017b, Vijay2017,Ma2017, Kim2017, Hsieh2017, Pretko2017, Slagle2017, Devakul2017, Petrova2017, Prem2017a, Prem2017c, Schmitz2018, Ma2018, Slagle2018a, Slagle2018b,Bulmash2018, Pai2018a, Ma2018a, Song2018}. Fractonic models are typically 3D models characterized by excitations which are constrained to move on sub-dimensional manifolds. The two most studied fracton stabilizer codes are X-cube (Type I)\cite{Vijay2015, Vijay2016} and Haah's cubic code (Type II) \cite{Haah2011}, henceforth referred to as ``Haah's code,'' though there are many others. Continuous (gapless) models exhibiting the same phenomenology have also been found and explored in the work by M. Pretko and others on tensor gauge theories \cite{Pretko2017a, Pretko2017b, Pretko2017, Pretko2018, Pai2018, Slagle2018a}. For a review on the state of fractons, see Ref.~\cite{Nandkishore2018}.

In this paper, we explore fracton stabilizer codes as their novelty, oddity and lack of full rotation symmetry demonstrates the need for a more general viewpoint on gauge theories beyond Yang-Mills theory. To this end, we introduce a mathematical construct which we term a {\it linear gauge structure}. This can be used to characterize both stabilizer codes and $U(1)$ gauge theory. This framework has the benefit of being independent of any Hamiltonian/Lagrangian, Hilbert space or algebraic field (as in $\mb R$ versus $\mb F_2$). We can also use it to prove some powerful results, most notably a generalization of charge conservation (or {\it BrLE rules} 
as we term them; see Section~\ref{sec:ex}) which is responsible for the fractonic behavior of the models we consider. We also discuss topological features as captured by a certain subset of gauge structures we term{ \it symplectic gauge structures}, which includes all stabilizer codes. For stabilizer codes, topological features are characterized by the logical operators which act non-trivially on the code space. We prove that these operators are connected to a set of {\it constraints}, or product of stabilizers which multiply to the identity. These special constraints are topological as they are no longer constraints if we change the topology of the system. We find an explicit construction for logical operators such that they form in the intersection of the boundary with these topological constraints. 

The machinery of gauge structures naturally leads to a means of mapping discrete stabilizer codes to continuum models and we apply this to all examples introduced below. The toric code maps to $U(1)$ gauge theory, whereas the analogous Maxwell equations for the X-cube model have the form
\begin{subequations}
\begin{align}
j_0 =& (\partial_1^2 \partial_2^2 + \partial_2^2 \partial_3^2 + \partial_3^2  \partial_1^2) A_0, \label{eq:XC0}\\
j_1 =& \partial_2^2(A_3-A_1) - \partial^2_3(A_1-A_2), \label{eq:XC1}\\
j_2 =& \partial_3^2(A_1- A_2) - \partial_1^2(A_2-A_3),\label{eq:XC2} \\
j_3=& \partial_1^2(A_2-A_3) -\partial_2^2(A_3-A_1),\label{eq:XC3}
\end{align}
\end{subequations}
where $\partial_i$ is the derivative in the $i^{th}$ coordinate direction and $j_i$ and $A_i$ are the currents/charges and potential fields, respectively. Each term on the LHS of Eq. \ref{eq:XC0} has two derivatives in different coordinates and as a consequence,
\begin{align}
\int dx_i dx_j j_0 =0,
\end{align}
for $i\neq j$. For similar reasons, Eqs \eqref{eq:XC1},\eqref{eq:XC2} and \eqref{eq:XC3} imply,
\begin{align}
\int dx_i dx_{j} j_{k} =0,
\end{align}
where $i \neq j \neq k$. Furthermore, we also find that $j_1 + j_2 +j_3 =0$. We argue these conditions are the conservation laws that are responsible for fractonic behavior in this continuum model. This is also similar to the ``hollow'' gauge theory first discussed in \cite{Ma2018a}.

For Haah's code, the Maxwell equations are
\begin{subequations}\label{eq:HaahMax}
\begin{align}
j_0 =& \left(\left(\partial_{\text{mix}}^2\right)^2 - \partial_{[111]}^2\right) A_0, \\
j_1 =&  \left(\left(\partial_{\text{mix}}^2\right)^2 - \partial_{[111]}^2\right) A_1,
\end{align}
\end{subequations}
where $\partial_{[111]}= \partial_1 + \partial_2 + \partial_3$ and $\partial_{\text{mix}}^2= \partial_1 \partial_2 + \partial_2 \partial_3 + \partial_1 \partial_3 = \hlf\left( \partial_{[111]}^2- \partial^2 \right)$. Again, this model is similar to the model discussed in \cite{Bulmash2018}. As we show, the only reasonable way to write the conservation law for these charges is by noting that
\begin{align}\label{eq:haahcon}
\int d^3x j_i a_i =0,
\end{align}
where $a_i$ is any field which satisfies $\partial_{\text{mix}}^2 a_i= \partial_{[111]}a_i =0$.  Clearly Eq. \eqref{eq:haahcon} is true for the fields as given by Eqs. \eqref{eq:HaahMax}. One can think of $a_i$ as defining a kind of moment for the charges of this model. Though it is not clear this model is completely fractonic, it does lend some intuition for better understanding Haah's code. In particular, we discuss a possible connection between Haah's code and coupled layers of the 2D classical model introduced by Newman and Moore \cite{Newman1999}. As a final comment, we discuss our interpretation of gauge structures. As there is no Hamiltonian or Lagrangian, a gauge structure is agnostic to matter and thus the usual interpretation of gauge theory as a phase invariance of matter fields no longer applies. Instead, we suggest the gauge redundancy can be interpreted as a necessary by arbitrary choice in the tensor product structure (TPS) of the Hilbert space. 

The structure of the paper is as follows: we first introduce all of the necessary background and theory of stabilizer codes in Section \ref{sec:stabgen} including some new results on the relation between constraints among the stabilizers and the excitations and logical operators. We then abstract these notions to define a linear gauge structure in Section \ref{sec:gauge} and prove some important general results. This machinery is used in Section \ref{sec:connecting} as well as the polynomial ring formalism for translationally symmetric stabilizer codes \cite{Haah2013} to form the continuum versions of the toric code and fracton codes introduced above. We then revisit Haah's code using the continuum version to gain some intuition. Finally, we make some concluding remarks and discuss our interpretation in Section \ref{sec:discus}.       

\subsection{Notation}

We use small Roman letters for operators and elements of sets, capital Roman letters for finite simple sets and at times continuum fields (where there should be no confusion between the two uses), and script Roman letters for linear or algebraic spaces. All maps use Greek letters. The composition of maps $\alpha$ and $\beta$ is written as $\alpha \beta$, $\star$ represents the pullback, $\iota_{\mc A}$ the inclusion map for the space $\mc A$ and $\id_{\mc A}$ the identity map on $\mc A$. $\dim \mc A$ is the dimension of $\mc A$. For the sake of readability, we use square-brackets for the image and pre-image i.e, $\alpha[A]$ and $\alpha^{-1}[A]$, respectively. We also use $+$ for the symmetric difference\footnote{Recall the symmetric difference is give by $A + B = A \cup B - A\cap B$.}, where there should be no confusion when used for sets, $\wp(A)$ for the power set of $A$ or the set of all subsets of $A$ and $|A|$ for the number of elements in the finite set $A$. 

\section{Stabilizer Codes} \label{sec:stabgen} 

In this section, we elaborate on stabilizer codes and their relevant features, following and extending results from \cite{Gottesman1997, Haah2011, Haah2013, Terhal2015}. Such details are required for a full understanding of the connection to the aforementioned gauge structure. Those already familiar with such details may skip to the next section though some results toward the end are new to the best of our knowledge.

We consider primarily spin-$\hlf$ or qubit systems, but it should be obvious that all results easily generalize to arbitrary spin or qudit systems.  

\subsection{The Pauli Vector Space and Pauli Boolean Algebra}

We start with a system of qubits typically arranged on the vertices, edges or other geometric unit of some graph (though this is not necessary). 
Also, we typically require each vertex has the same number of qubits, $n$ (though again, this is not necessary). 
 $Q$ denote the set of all qubits with $N =|Q|$ being the total number of qubits. Thus the Hilbert space is $\mc H \simeq \mb C_2^{\otimes N}$ with dimension $ 2^N$. 
Initially, we make no assumption on the symmetry of the system, though we only consider examples with translation symmetry.

Let $\mc P$ be the set of all products of single-qubit X-, Y- and Z-type Pauli operators (simply referred to as Pauli operators) acting on the qubits of our system, modulo any phase of $\pm1, \pm i$. We use $x_q,y_q$ and $z_q$ for the single-qubit X-, Y-, and Z-type Pauli operators for qubit $q \in Q$. For any member $f \in \mc P$ we define the {\it support} as
\begin{align}
\supp(f) =\{q \in Q: f_q \neq \id_{\mb C_2} \},
\end{align}
where $f_q$ is the part of $f$ acting in the product space associated with $q\in Q$. Likewise for any $F\subseteq \mc P$, $\supp(F) = \bigcup_{f \in F} \supp(f)$. Ignoring the phase allows us to treat $\mc P$ as a vector space over the field of two elements, $\mb F_2$, where addition of two Pauli's is  given by their product and scalar multiplication corresponds to the power. Since $f^2 \propto \id_{\mc H}$, $\mb F_2$ is the appropriate field. But as we have modded out the phase, it seems we have thrown away the commutation relations. We can recover this information by introducing the symplectic form $\lambda: \mc P\times \mc P \to \mb F_2$, which encodes the commutation relations via
\begin{align}\label{eq:lamdef1}
(f,g) \mapsto \lambda(f,g)= \begin{cases}
0 & \text{ if $f$ and $g$ commute,}\\
1 & \text{otherwise.}
\end{cases}
\end{align}
As any two Pauli's either commute or anti-commute, this encodes all commutation relations. It is symplectic since it is a bi-linear non-degenerate form (a concept discussed below) which satisfies $\lambda(f,g)= -\lambda(g,f)$. We have a natural basis for $\mc P$, namely the set of all single qubit $X$- and $Z$-type Pauli's (as $y_q \propto x_q z_q$). As this is a basis, we conclude $\dim \mc P = 2 N$. Note this basis also has a special property that it divides into two subsets $\{x_q : q \in Q\}$ and $\{z_q: q \in Q\}$ such that $\lambda(x_q, x_p) = \lambda(z_q,z_p)=0$ and $\lambda(x_q, z_p) = \delta_{qp}$. This form of basis is general for a symplectic vector space. That is, given any basis, one can always form a {\it canonical basis},  $\{f_i\}_{i\in \mb Z_N}\cup \{g_i\}_{i\in \mb Z_N}$ such that $\lambda(f_i, f_j) = \lambda(g_i,g_j)=0$ and $\lambda(f_i,g_j) = \delta_{ij}$. As a corollary, the maximum number of independent, mutually commuting operators is $N$. 

To make the connection with continuum gauge structures, it is worth mapping $\mc P$ onto a different $\mb F_2$ vector space. Consider that $\wp (Q)^2 = \wp ( Q)_X \times \wp (Q)_Z$ is another $\mb F_2$ vector space with addition given by entry-wise symmetric difference such that for $F=(F_X, F_Z), G= (G_X, G_Z) \in \wp (Q)^2$, $F + G = (F_X+ G_X, F_Z + G_Z)$. It should be obvious that $\wp (Q)^2$ is isomorphic to $\mc P$, where for $f \in \mc P$ and $F \in \wp(Q)^2$ such that $f \mapsto F$, the first factor ($F_X$) represents the set of all X-type single Pauli operators forming $f$ and the second factor ($F_Z$) is the set of all Z-type Pauli operators forming $f$. In this language, we can also express $\lambda$ as
\begin{align}\label{eq:lamdef2}
\lambda(F,G) = |F_X \cap G_Z| + |G_X \cap F_Z|,
\end{align}
where the sum is mod 2. Although it is clear this reproduces $\lambda$ as first defined, it is worth pointing out that $\lambda$ is linear because intersection distributes over symmetric difference, owing to the fact that $\wp (Q)^2$ is a Boolean algebra with intersection as the product. Finally, we can express the support function as 
\begin{align}
\supp(F) = F_X\cup F_Z.
\end{align}

\subsection{The Stabilizer Set, Group and Boolean Algebra}\label{sec:stab}

We now define the stabilizer set and its associated group and Boolean algebra.

\begin{definition} \label{def:stab}
 The set $ S \subset \mc P$ (or $\wp (Q)^2$ if you like) is a stabilizer set if at minimum it satisfies the following:
\begin{enumerate}
\item $ \lambda[S \times S]=\{0\}$, i.e. it is composed of mutually commuting operators and
\item $ id_\mc H \notin  S$.
\end{enumerate}

Though these conditions are technically sufficient, it is often best to include a few additional restrictions:
\begin{enumerate}
\setcounter{enumi}{2}
\item $|S| \geq N$, i.e. there are at least as many stabilizers as there are qubits in the system,
 \item $\supp( S) =Q$, i.e. every qubit is acted upon nontrivially by at least one member of $S$ and
\item Any symmetry of the {\it stabilizer group} is also a symmetry of $S$.
\end{enumerate}
\end{definition}

The last requirement involves the stabilizer group denoted by $\mc G$ and defined as
\begin{align}
\mc G= \{ \prod_{s \in F} s \mod(\pm 1, \pm i) : F\in \wp (S)\}.
\end{align}
Thus $\mc G \subset \mc P$ is the set of all Pauli operators generated by taking products of members from $S$, modulo any phase. We continue to use the term group as this is the usual terminology, but for our purposes, it is much more useful to think of it as a subspace of $\mc P$ which is spanned by all of $S$. We refer to members of $S$ as stabilizers and members of $\mc G$ as stabilizer group elements.\footnote{This is not the usual terminology as all members of $\mc G$ are typically called stabilizers. But for the purposes of this discussion, it is better to distinguish only members of $S$ as the stabilizers.} To further explain requirement 5 of Definition \ref{def:stab}, we let the set of symmetries of $\mc G$ be denoted by $\Pi$ which contains members $\pi : \mc P \to \mc P$ (again or $\wp (Q)^2 \to \wp(Q)^2$ ) which are isomorphisms that preserves all of $\mc G$, i.e. $\pi[\mc G] = \mc G$. So our last requirement means that if $\pi \in \Pi$, then $\pi[S] = S$. This may seem like an obvious requirement, but in fact, it is pivotal to this discussion. It has the effect of forcing certain members to be contained in $ S$ and necessitates the exists of constraints as discussed below. Furthermore, it is a reasonable requirement from a physical viewpoint. We do not want the stabilizer set to treat some parts of the graph (our discrete proxy for space) differently than others. We also note this requirement is not ideal in the sense that we don't want to make symmetry an essential ingredient. What we want is a kind of homogeneity of space, which is not necessarily the same as preserving symmetry. We discuss cases where neither $\mc G$ nor $S$ have a symmetry, but we still want all parts of the graph to be treated on equal footing. 

The stabilizer code is then defined as the subspace $\mc H_{\text{code}} \subseteq\mc H$ such that
\begin{align}\label{eq:code}
\mc H_{\text{code}}= \{ \ket{\psi} \in \mc H : s\ket{\psi} = \ket{\psi} \text{ for all } s\in S\}.
\end{align}
%
This can be conveniently understood as the ground space of the Hamiltonian, 
\begin{align}\label{eq:ham}
H= - \sum_{s \in S} s,
\end{align}
where we implicitly choose the phase of $s$ such that it is Hermitian. However, we emphasis that nothing that follows requires this to be the Hamiltonian. 

We now introduce our examples. The first  is the toric code in $2$ and $3$ dimensions. In complete generality, we consider any tessellation of the $2$- or $3$-torus. To each edge, we assign one qubit and the stabilizer set is given by
\begin{align}
S_{TC} = \{a_v, b_f: v \text{ vertices and } f \text{ faces}\}.
\end{align}
$b_f= \prod_{i \in f} z_i$, where $i$ indexes the qubits about the face $f$ and $a_v = \prod_{i @ v} x_i$, where $i$ indexes the edge qubits coordinated to the vertex $v$. Here we see the limitations of the symmetry requirement. Take the $2$-torus for example. We can remove exactly one face and one vertex stabilizer from $S_{\text{TC}}$ and $\mc G_{TC}$ would be unaffected (see discussion in Section \ref{sec:const}). If the tessellation is the square lattice, then the symmetry requirement forbids removing these stabilizers. However, if the tessellation is a random or aperiodic tiling, there may not be any symmetry to forbid removing these stabilizers.  However, by our intuitive notion of treating all parts of the graph equally, we would include all face and all vertex stabilizers. This point aside, we henceforth only consider the toric code on the square and cubic lattices where faces are referred to as plaquettes.

The first fracton model we consider is the X-cube model.  Though this has been generalized to other geometric structures \cite{Slagle2018b}, we consider the original defined on a cubic lattice with a qubit associate to each edge and  
\begin{align}
S_{\text{XC}} = \{a_v^{\hat e_1} , a_v^{\hat e_2}, a_v^{\hat e_3}, b_c : v \text{ vertices and } c \text{ primitive cubes}\}.
\label{eq:HXC}
\end{align}
$a_v^{\hat e_j} = \prod_{ i @ v_j} z_i$, where $i$ indexes the qubits coordinated to vertex $v$ and for edges confined to the plane normal to $\hat e_j$ and $b_c = \prod_{i \in c} x_i$, where $i$ indexes the qubits about the primitive cube $c$. This is the paradigm example of a Type I fracton model. Finally, we also consider the paradigm example of a Type II fracton model, Haah's code, which is defined on the cubic lattice with two qubits per vertex and 
\begin{align}
S_{\text{Haah}}= \{g_v^X,  g^Z_v: v \text{ vertices}\},
\end{align}
where $g_v^{X,Z}$ are defined in Fig. \ref{haah0}.

\begin{figure}[t]
\centering
\includegraphics[scale=1.2]{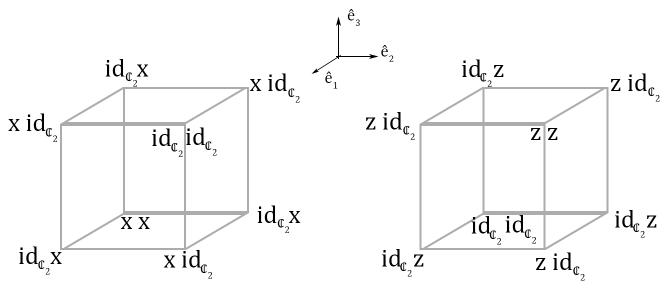}
\caption{Definition of the two stabilizer types in Haah's code were $g_v^x$ is on the left and $g_v^z$ on the right. All others are related by translation.}\label{haah0}
\end{figure}
All examples presented here are Calderbank, Shor and Steane (CSS) stabilizer codes \cite{Terhal2015}. That is, stabilizers are composed of either all $X$-type or $Z$-type operators. Even though our examples are of the CSS form, we do not require it for what follows. 

Similar to the discussion of the Pauli space and Boolean algebra, we consider the algebra $(\wp (S), +, \cap)$ as generated by the stabilizer set. This is also equipped with a two-form $\omega: \wp(S)\times \wp(S) \to \mb F_2$ such that
\begin{align}
(A,B) \mapsto \omega(A,B) =|A \cap B| \mod 2.
\end{align}
Unlike $\lambda$, this is not a symplectic form, nor is it an inner product as $\omega(A,A)=0$ does not imply $A=0$. However, it is bilinear (a not so obvious property) and non-degenerate. If we were to map the space $\wp (S)$ onto $\{0,1\}^{|S|}$, i.e. the set of binary strings of length $|S|$ with standard addition mod 2, one recognizes $\omega$ as the binary ``dot product,'' so this definition is natural. $\omega$ can be used to construct vectors using the special singleton basis $\{\{s\}: s\in S\}$ via 
\begin{align} \label{eq:recon}
A= \sum_s \omega(A, \{s\}) \{s\}.
\end{align}
It was shown in Ref. \cite{Bernhard2009} that such a 2-form is {\it invertible}\footnote{In the reference, it is referred to as the reconstruction identity.} which states in simple terms that if we can calculate $\omega(A, B)$ for all $B\in \wp(S)$, then we can construct $A$ uniquely and Eq. \eqref{eq:recon} is one way to do so. We give a rigorous definition of invertibility in Section \ref{sec:gauge}.

\subsection{Constraints}\label{sec:const}

To relate members of $\wp(S)$ back to operators, we use the map $\phi:\wp(S) \to \mc P$ such that
\begin{align}
A \mapsto \phi(A)= \prod_{s\in A}s. 
\end{align}
Note that $\im \phi = \mc G$. It should be clear that $\phi$ is linear, however, it is neither surjective nor injective. We call the space $\ker \phi$  the {\it constraint space} and its members constraints for which $C \in \ker\phi$ implies $\prod_{s\in C} s \propto \id_{\mc H}$.  In all cases considered here, some of these constraints have special meaning with respects to the topology of the system. To characterizes these constraints, 
we define a {\it change of boundary} map.
\begin{definition}
A change of boundary is a bijective map $\beta: S\to S'$ where both $S$ and $S'$ are stabilizer sets on the same set of qubits $Q$ such that:
\begin{enumerate}
\item The logical subspace for $S'$ is trivial\footnote{See Section \ref{sec:lop} for definition. We consider this map to be something like changing the topology of the system, though the Hilbert space and number of qubits has not changed. We take the perspective that it is the stabilizers which encode the topology as encapsulated by the gauge homology/logical subspace.},
\item $\mc G \cap \mc G'$ is generated by $S\cap S'$ and
\item for all $s_1, s_2 \in S$, $\lambda\left(\left(\beta(s_1) + s_1\right), \left(\beta(s_2) + s_2\right)\right) =0$.
\end{enumerate} 
\end{definition}
In the context of topological stabilizer codes, we often refer to $S$ as having {\it periodic boundary conditions} (pbc) and $S'$ as having  {\it open boundary conditions} (obc). The first condition says the change of boundary removes all topology. The last two are the purposed generalization to the notion that the change of boundary alters only a sub-extensive number of stabilizers i.e. stabilizers should only be altered along a $d-1$ boundary in $d$ dimensions. Such a condition would only make sense on a lattice and we want to keep the discussion as general as possible. Condition 2 can be interpreted as the two stabilizer groups $\mc G$ and $\mc G'$ are the same away from the boundary, where boundary refers to the support of altered stabilizers. As for the last condition, $\beta(s) + s$ represents the part of the operator which is altered by the change of boundary and we want this change to be ``minimal'' in some sense. 
To demonstrate how condition 3 facilitates this, consider the $d=2$ toric code where we start forming a change of boundary by removing one X-type Pauli from some vertex stabilizer. This process is depicted in Fig. \ref{minbd}. The new T-shaped operator anti-commutes with two plaquette stabilizers and we have to alter them to form $S'$. However, if we were to remove the Z-Pauli operator on that same qubit for either of these two plaquettes, not only would we not change the commutativity of the resulting U-shaped operators with our T-shaped operator, but we would create even more places where the altered operators anti-commute with stabilizers. So a more ``minimal'' choice is to remove Z-type operators which do not overlap with the X-type operator already removed. If this is not convincing, condition 3 as well as condition 2 are sufficient to prove an important result below, so we take them as reasonable.  Note that condition 3 can be written as $\lambda(\beta(s_1), s_2) = \lambda(s_1, \beta(s_2))$ for all $s_1, s_2 \in S$, where we use linearity and the fact that both $S$ and $S'$ are proper stabilizer sets. Also, $\beta$ is entirely symmetric with respects to $S$ and $S'$ except for condition 1 as one might expect. 

\begin{figure}[t]
\centering
\includegraphics[scale=.4]{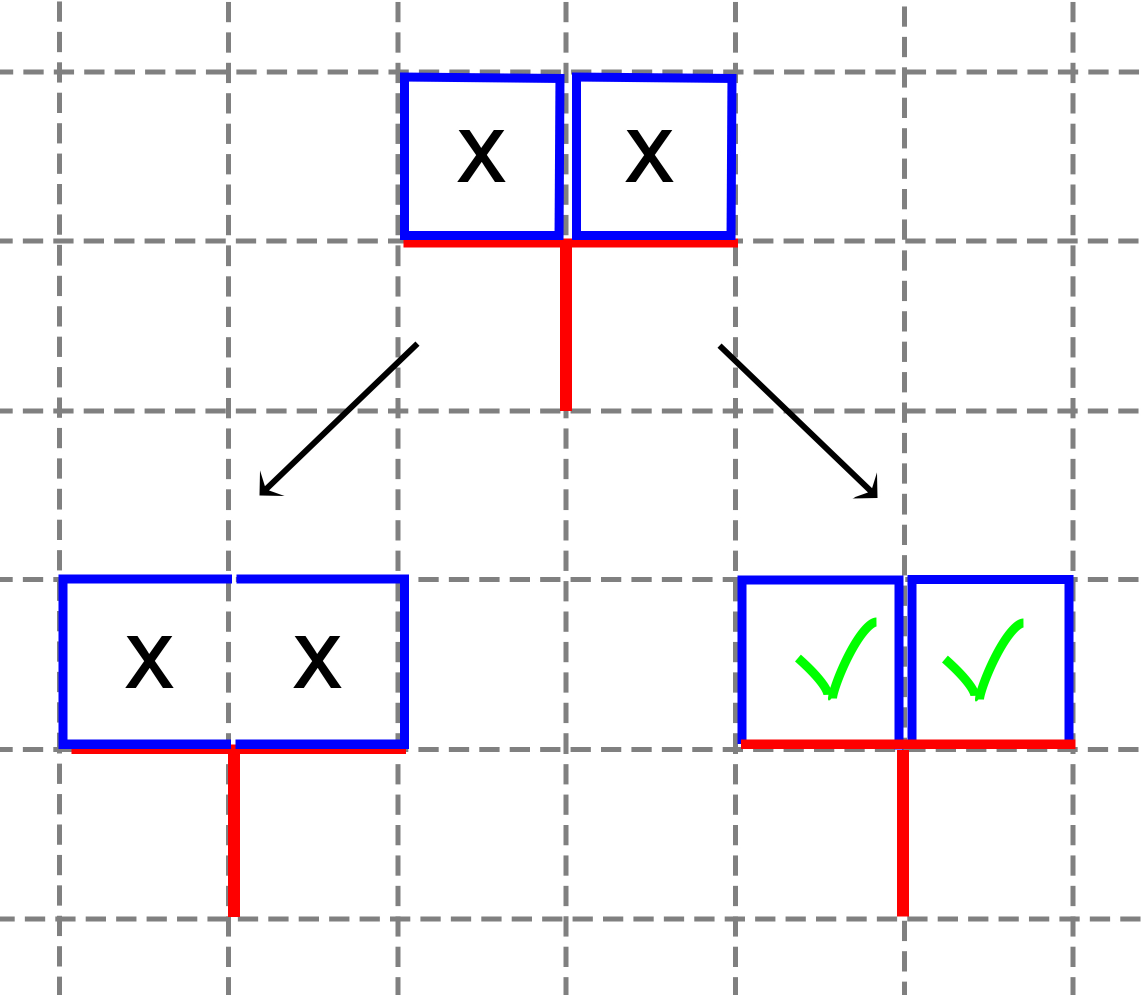}
\caption{Demonstration of a minimal change of boundary for the $d=2$ toric code. Starting with the modified vertex stabilizer at the top, we can see the right modification to the nearby plaquettes should result in fewer deletions than the modifications on the left. We use this as justification for condition 3 on a change of boundary map. }\label{minbd}
\end{figure}

We now intend to capture how the constraint space is altered by the inverse change of boundary $\beta^{-1}$, but both the image and pre-image under $\beta^{-1}$ of a constraint is not necessarily a constraint. We thus construct a constraint-space map $\gamma: \ker \phi' \to \ker \phi$ using $\beta^{-1}$ via the mapping 
\begin{align}
C' \mapsto \gamma(C')=\beta^{-1}[C'] + {\min_{\supp}}' \phi^{-1}\left[\{ \phi \beta^{-1} [C']\}\right].
\end{align}
The sequence of maps, $\phi^{-1}\left[\{ \phi \beta^{-1} [C']\}\right]$ deserves some explanation:  we first map $C' \in \ker\phi'$ into $\wp(S)$ via the image under the inverse change of boundary map $\beta^{-1}$, then map the resulting set to the operator it generates via $\phi$. As we have possibly changed support in going from $S'$ to $S$, it is possible $\beta^{-1}[C']$ does not map to the identity under $\phi$. We then apply the pre-image of that operator under $\phi$ to find all sets which generate it and choose the one with minimal support. There may be some ambiguity in which set has minimal support, but we assume in the general case the choice can be made. As any member $F \in \phi^{-1}\left[\{ \phi \beta^{-1} [C']\}\right]$  generates the same operator as $\beta^{-1} (C')$ under $\phi$, $F + \beta^{-1}[C] \in \ker \phi$ by linearity. We then choose the ``smallest'' $F$. The prime on the minimum denotes that we restrict the minimization to sets such that all members are unaltered by the change of boundary, i.e. sets $F$ such that $\beta[F] =F$. We prove this restriction is always possible in the following proposition.
\begin{proposition}
For all change of boundary $\beta: S \to S'$ and $C' \in \ker \phi'$, there exists an $F \in \phi^{-1}\left[\{ \phi \beta^{-1} [C']\}\right]$ such that $\beta[F] =F$. 
\end{proposition}

\begin{proof}
Let $\beta: S \to S'$ be a change of boundary and $C' \in \ker \phi'$ and consider the following commutation with any $\beta(s_2) \in S'$:
\begin{align} \label{eq:pf1}
\lambda(\phi \beta^{-1} [C'], \beta(s_2)) =& \lambda\left( \prod_{s_1 \in \beta^{-1}[C']} s_1, \beta(s_2) \right)  = \sum_{s_1 \in \beta^{-1}[C']} \lambda( s_1, \beta(s_2) ) \nonumber \\
=& \sum_{s_1 \in \beta^{-1}[C']} \lambda(\beta(s_1), s_2) =\lambda\left( \prod_{s' \in C'} s', s_2 \right) =\lambda(\phi'(C'), s_2) =0,
\end{align}
where we have used the definition of $\phi$ and $\phi'$, linearity, and condition 3 on $\beta$. Thus $\phi \beta^{-1} [C']$ commutes with all $S'$ and since the logical subspace for $S'$ is trivial (again see definition in Section \ref{sec:lop}), this implies $\phi \beta^{-1} [C'] \in \mc G'$. Trivially, $\phi \beta^{-1} [C'] \in \mc G$. By condition 2 on $\beta$, this implies there exists a set $F \in \wp(S\cap S')$ such that $\phi \beta^{-1} [C'] = \phi(F) = \phi'(F)$. Therefore, $F \in \phi^{-1}\left[\{ \phi \beta^{-1} [C']\}\right]$ where $\beta[F] =F$. 
\end{proof}
 If $\beta^{-1}[C']$ is a constraint, then ${\min_{\supp}}' \phi^{-1}\left[\{ \phi \beta^{-1} [C']\}\right] = \emptyset$ and $\gamma(C') = \beta[C']$ as expected.
 
Constraints which can be generated via this map are call non-topological or {\it trivial} constraints (topologically trivial, as we shall argue they have other consequences). However, there may be some constraints which cannot be generated this way for any change of boundary.  Such constraints are members of the{ \it topological constraint space} defined as\footnote{It is possible that $\gamma$ is not linear due to our use of minimum support for selecting the ``counter term''. If it happens that $\gamma$ is not linear, we replace $\im \gamma$ with $\spn(\im \gamma)$, i.e. we complete the image so that it is a subspace of $\ker \phi$.}
\begin{align}
\mc T = \cap_\gamma \ker\phi / \im \gamma,
\end{align}
where the intersection over $\gamma$ is for  constraint space maps as induced by some change of boundary. By this definition, a topological constraint is really an equivalence class of constraints where all members are equivalent modulo some augmentation or deformation by (sum with) some trivial constraint. We reserve a deeper discussion of topological constraints for future work, however some discussion of their importance can be found in \cite{Schmitz2018}. 

For concreteness, we consider the constraints for the models discussed above. A more detailed discussion for these examples can be found in \cite{Schmitz2018}. The constraint space for the $d=2$ toric code is generated by the vertex constraint $\{a_v: \text{for all vertices }  v\}$ and the plaquette constraint $\{b_p: \text{ for all plaquettes } p\}$. It should be clear that these constraints are not present for the toric code with obc and so all constraints are topological. Going to $d=3$, we still have both these constraints, however there are more constraints among the plaquette stabilizers. The set of all plaquettes about any elementary cube forms a constraint which is clearly not topological. These cube constraints generate an entire space of trivial constraints which is present even with obc. Such a space of constraints can be characterized as all sets of plaquettes which form closed membranes (or terminate on a ``ragged'' boundary if it exists).  A constraint on the 3-torus that is not generated by these cubes is the set of plaquette operators which forms a plane wrapping around the system in any two of the three directions. This suggests the dimension of $\mc T_{\text{TC}_3}$ is four (one for the set of all vertices, and one for each wrapping plane), but this is incorrect by our definition. In obc, the sum over all cube constraints is exactly the sum of the three wrapping plane constraints in pbc, which is to say under $\gamma$, the sum of all cubes maps to the sum of the three wrapping planes and is therefore trivial. Thus $\dim \mc T_{\text{TC}_3}=3$ due to this linear dependence among the three wrapping plane constraints.

As discussed in \cite{Ma2017, Ma2018, Schmitz2018}, the X-cube model contains a sub-extensive number of topological constraints and an extensive number of trivial constraints. Every 2D layer of cubes in any of the three directions forms a topological constraint, however, these are not all independent. In particular, the sum of all layers perpendicular to any one direction (say $\hat e_1$) is equivalent to the same sum in any other direction (say $\hat e_2$). As for the vertex stabilizers, there is an extensive number of trivial constraints generated by the collection of sets, one for each vertex, where each set contains the three stabilizers associated to that vertex. But unlike the $d=3$ toric code, these vertex constraints are far more ``tame.'' Each plane of vertex stabilizers is then a topological constraint as they only exist on the 3-torus and cannot be deformed into one another by the trivial constraints. However, the sum of all such layers in all three directions is just the sum over all vertex stabilizers, which is a trivial constraint and does not  contribute to $\dim \mc T_{\text{XC}}$. Taking all this into account, we find that $\dim \mc T_{\text{XC}}= 6L-3$ for a three torus of linear size $L$. The sub-extensive scaling of this dimension is another common feature of fractonic models \cite{Ma2018, Schmitz2018}.

As for Haah's code, all constraints must be topological as simple counting tells us that $|S_{\text{Haah}}| = N$ and the logical subspace is non-trivial (see Section \ref{sec:lop}). But as this dimension is a complicated function of the linear system size \cite{Haah2011}, so too is the number of constraints.  We shall make some progress on understanding these constraint, but this is best done using the polynomial ring formalism \cite{Haah2013} which we introduce in Section \ref{sec:connecting}, so we delay the discussion until then.  

\subsection{Excitations}\label{sec:ex}

We typically think of excitations as states with energies above the ground space, but as we wish to remain Hamiltonian independent, we define them via a map $\psi: \mc P \to \wp(S)$ such that 
\begin{align}
f \mapsto \psi(f)= \{s \in S: \lambda(s,f)=1\}.
\end{align}
 We refer to this as the {\it syndrome map} and members of its image syndromes due to the connection with error syndromes in the context of error correction. Individual members of a syndrome are called fundamental excitations and subsets composite excitations. This map tells us which stabilizers ``flip'' (change eigenstate) under the action of $f$ on a member of the stabilizers code as defined in Eq.~\eqref{eq:code}, but it also has an important connection to $\phi$ and $\omega$: 

\begin{proposition}[Braiding Condition]\label{prop:stabgauge}
For all $A \in \wp(S)$ and $f \in \mc P$,
\begin{align}\label{eq:stabgauge}
\lambda(\phi(A), f) = \omega(A, \psi(f)).
\end{align}
\end{proposition}

\begin{proof}
We start with the left-hand side for which
\begin{align}
\lambda(\phi(A), f) = \lambda(\prod_{s\in A}s, f) = \sum_{s \in A} \lambda(s, f),
\end{align}
where we have used the definition of $\phi$ and the linearity of $\lambda$. For each term in the sum, $\lambda(s, f)=1$ if and only if $s\in \psi(f)$ and zero otherwise so that 
\begin{align}
\lambda(\phi(A), f) = \sum_{s \in A \cap \psi(f)} 1= |A \cap \psi(f)| = \omega(A, \psi(f)).
\end{align}
\end{proof} 

As we show in general, $\psi$ is the unique map with this property. We argue that this relation implicitly encodes all abelian braiding processes for the stabilizer code, hence we refer to it as the {\it braiding condition}. To see this, consider what a full braid means in this context. Starting with some stabilizer state, the entire braid is described by an operator whose action maintains the state up to a possibly different phase. In the stabilizer language, this implies the operator is a stabilizer group element in which case we can write the operator as $\phi(A)$ for some $A \in \wp(S)$. Furthermore, we can achieve any stabilizer eigenstate by acting with some Pauli operator $f \in \mc P$ on a member of the stabilizer code which we label as $\ket{\psi(0)} \in \mc H_{\text{code}}$, i.e. $\ket{\psi(f)} = f \ket{\psi(0)}$. So applying our full braid to this state,
\begin{align}
\phi(A) \ket{\psi(f)} =  \phi(A) f \ket{\psi(0)}= (-1)^{\lambda(\phi(A), f)} f \phi(A)\ket{\psi(0)} =(-1)^{\omega(A, \psi(f))} \ket{\psi(f)}.
\end{align}
where we have used Proposition~\ref{prop:stabgauge} and the definition of the stabilizers state $\ket{\psi(0)}$. So the left-hand side of Proposition~\ref{prop:stabgauge} describes the phase of the braid and the right-hand side describes the geometry, i.e. counts the number of excitations which are ``encircled'' by the braid (as defined by $A$). It may see strange to call this a braid, as braiding is typically thought of as a process of moving excitations around one another. $\phi(A)$ is often not a local operator, but it can always be decomposed into a product of local operators which when applied sequentially, has the appearance of moving excitations. It is important to note that our braiding condition does not define equivalence classes among members of $\im \phi$. Defining such equivalence classes would be an interesting topic of future work though a hint can be found in Ref.~\cite{Schmitz2018} and the discussion surrounding {\it non-local surface stabilizers} and emergent Gauss laws.
  
 Another important consequence of Proposition~\ref{prop:stabgauge} is the {\it Braiding Law for Excitations} or BrLE rules, which is stated as followings:
\begin{theorem}[Stabilizer BrLE Rules]\label{prop:stabcon}
$J$ is a syndrome i.e. $J \in \im\psi$ if and only if $ \omega(C, J)= 0$ for all constraints $C \in \ker \phi$. 
\end{theorem}
Simply put, any pattern of excitations (errors) can be realized via the action of some Pauli operator $f$ if and only if that pattern overlaps with every constraint an even number of times. We prove Theorem~\ref{prop:stabcon} in full generality in Section \ref{sec:gauge}. It has been realized in the literature \cite{Terhal2015, wenbook} that one way is true ($\Rightarrow$) as it is trivial: since a constraint maps to the identity, its product must commute with everything and so the number of members which anti-commute with some Pauli operator must be even. We have not found the converse stated in full generality anywhere else, though a version for CSS codes is discussed in Ref.~\cite{Kubica2018}. The fact that this is both necessary and sufficient can be a very powerful tool as it requires that any general statement about the excitations must have a realization in the structure of the constraints. For example, excitations forming closed loops in the $d=3$ toric code is a direct consequence of the extensive number of trivial constraints. Any syndrome must contain an even number of plaquettes of any elementary cube. This forms a string-like composite excitation which must exit every volume it enters and thus forms a closed loop. Likewise, the fractonic behavior of the X-cube model is a consequence of any syndrome having to satisfy the BrLE rules for multiple intersecting planar constraints. As for Haah's code, the fact that Theorem \ref{prop:stabcon} is necessary and sufficient means that the fractal-like nature of the excitations (see Ref.\cite{Haah2011}) must be a consequence of the so-far undiscovered constraints and we argue this is the case in Section \ref{sec:connecting}. Similar to definitions in Ref. \cite{Song2018}, this suggests that fractonic behavior is a consequences of excitations having to satisfy the BrLE rules for multiple {\it intersecting} constraints, whereby moving a fracton in any direction violates the BrLE rules and is thus not allowed. We plan to make this notion more precise in future work.

It is clear that the BrLE rules represent the statement of conservation of $\mb F_2$ (or $\mb Z_2$) charge. Thus we can loosely define topological charge conservation as those conservation rules which are enforced by the topological constraints. As we argue, this notion of conservation of charge is a general consequence of Eq. \eqref{eq:stabgauge} and the gauge structure given by $\phi$, $\lambda$ and $\omega$.

\subsection{Logical Operators}\label{sec:lop}

When constraints are present, it is often the case that there are too few independent stabilizers to form exactly one half of a canonical basis. Thus there are more operators in $\mc P$ which commute with all of $\mc G$. With our maps $\psi$ and $\phi$, we can conveniently capture all such operators in this {\it logical subspace} which we define as
\begin{align}
\mc H(S) = \ker \psi / \mc G = \ker \psi / \im \phi.
\end{align}
A member of $\mc H(S)$ is really an equivalence class of operators, each of which is called a logical operator as they collectively form a Pauli sub-algebra for information stored in $\mc H_{\text{code}}$ \cite{Terhal2015}. This immediately implies $\log_2 \dim \mc H_{\text{code}} =\hlf \dim \mc H(S)$ as only that many independent operators in $\mc H(S)$ can mutually commute. This space also defines the {\it code distance} as $d = | \min_\supp \mc H(S)|$, i.e. the size of the support of the smallest logical operator. Just as with excitations, logical operators can be related to constraints, or more specifically to topological constraints.

\begin{theorem}\label{thm:log}
 For every member $C \in \mc T$\footnote{What we really mean is a member of a member of $\mc T$, but for clarity, we continue to use this notation. Equality between two topological constraints then means they add to a trivial constraint. Likewise, $p_C \in \mc H(S)$ really means a member of a member of the logical subspace and equality of two such operators mean their sum is in $\im \phi$.} there exists a logical operator $p_C \in \mc H(S)$. Furthermore, for all $C_1, C_2 \in \mc T$, the logical operators can be chosen such that $p_{C_1}$ commutes with $p_{C_2}$ (i.e. they are all mutually commuting) and $p_{C_1} = p_{C_2}$ if and only if $C_1 = C_2$.  
\end{theorem}

This is essentially proven in two lemmas.

\begin{lemma}\label{lem2}
For all $C \in \ker \phi$ and all changes of boundary $\beta: S \to S'$, $\phi' \beta[C] \in \ker \psi$.
\end{lemma}

\begin{proof}
Let $C \in \ker \phi$ and $\beta: S \to S'$ be a change of boundary. Following the same method as Eq. \eqref{eq:pf1}, consider the commutation with any $s_2 \in S$, 
\begin{align} \label{eq:pf2}
\lambda(\phi' \beta[C], s_2) =& \lambda\left( \prod_{s_1 \in C} \beta(s_1), s_2 \right)  = \sum_{s_1 \in C} \lambda( \beta(s_1), s_2 ) =  \sum_{s_1 \in C} \lambda(s_1, \beta(s_2))\nonumber \\
 =&\lambda\left( \prod_{s_1 \in C} s_1, \beta(s_2) \right) =\lambda(\phi(C), \beta(s_2)) =0.
\end{align}
Therefore by definition, $\phi'\beta[C] \in \ker \psi$.  
\end{proof}

\begin{lemma} \label{lem3}
For all $C \in \ker \phi$ and all changes of boundary $\beta: S \to S'$, $C$ is trivial if and only if $\phi' \beta[C] \in \im \phi$.
\end{lemma}

\begin{proof}
Let $C \in \ker \phi$ and $\beta: S \to S'$ be any change of boundary. 

($\Rightarrow$) Suppose $C$ is trivial. This implies there exists a $C' \in \ker \phi'$ such that $\gamma(C')= C$ for $\gamma: \ker \phi' \to \ker \phi$  associated with $\beta$. Using linearity of $\phi'$ and $\beta$ and the definition of $\gamma$, this implies that 
\begin{align}
\phi' \beta[C]= \phi'\beta[\gamma(C')]= \phi'(C') + \phi' \beta\left[ {\min_{\supp}}' \phi^{-1}\left[\{ \phi \beta^{-1} [C']\}\right] \right] = \phi' \beta[F_{C'}],
\end{align}
where $F_{C'} ={\min_{\supp}}' \phi^{-1}\left[\{ \phi \beta^{-1} [C']\}\right] $. By the restriction on the minimization, $\beta[F_{C'}]= F_{C'}$. Therefore, $\phi' \beta[C] = \phi(F_{C'}) \in \im \phi$.

($\Leftarrow$) Now suppose  $\phi' \beta[C] \in \im \phi = \mc G$. Trivially, $\phi' \beta[C] \in \im \phi' = \mc G'$. So by condition 2 on $\beta$, there exists an $F \in \wp(S\cap S')$ (i.e. $\beta[F]=F$) such that $\phi' \beta[C] = \phi'(F)$. This implies that $C'=\beta[C] + F \in \ker \phi'$. Consider this constraint under the action of $\gamma$,
\begin{align}
\gamma(C') = \beta^{-1}[\beta[C] + F] + F_{C'} = C + F + F_{C'},
\end{align}
where we again define $F_{C'} ={\min_{\supp}}' \phi^{-1}\left[\{ \phi \beta^{-1} [C']\}\right] $. It is clear that $F \in  \phi^{-1}\left[\{ \phi \beta^{-1} [C']\}\right] =  \phi^{-1}\left[\{ \phi(F)\}\right]$ and $\beta[F]=F$ so that it satisfies the restriction on the minimization. Also note the set $ \phi^{-1}\left[\{ \phi(F)\}\right] = \phi^{-1}\left[\{ \phi' \beta[C]\}\right]$ is fixed (and well-defined by the hypothesis) regardless of our choice of $F$, so without loss of generality, we can choose $F= F_{C'}$. Therefore, $\gamma[C'] = C$ which implies $C$ is trivial.
\end{proof}

We now prove Theorem \ref{thm:log}.

\begin{proof}
Let $C\in \mc T$ and consider any change of boundary $\beta: S \to S'$. Let $p_C = \phi' \beta[C]$. By Lemma \ref{lem2}, we know that $p_C \in \ker \psi$, and by Lemma \ref{lem3}, we have $C \neq \emptyset$ if and only if $p_C \notin \im \phi$. therefore $p_C \in \ker\psi/ \im\phi = \mc H(S)$. Furthermore, since $S'$ is a proper stabilizer set, $\lambda(p_{C_1}, p_{C_2}) = \lambda(\phi' \beta[C_1], \phi'\beta[C_2]) =0$, for all $C_1, C_2 \in \mc T$. As for the final claim, note  $p_{C_1} = p_{C_2}$ if and only if $p_{C_1} + p_{C_2} = \phi'\beta\left(C_1+C_2\right)\in \im \phi$. By Lemma \ref{lem3} and $C_1 + C_2 \in \ker \phi$, this is true if and only if $C_1+C_2$ is trivial. Therefore, $p_{C_1} = p_{C_2}$ if and only if $C_1 =C_2$.
\end{proof}

This completes the proof of ideas first conjectured in Ref.~\cite{Schmitz2018} and can be summarized as logical operators form in the intersection of a topological constraint with the boundary. Using condition 1 for a change of boundary, we see that $\dim \ker \phi' =|S| - \dim \mc G' = |S|-N$. Using this, we also find that $\dim \mc T = \dim \ker \phi - \dim \ker \phi' = \dim \ker \phi- (|S| -N)$, where, importantly, we maintain the condition $|S| \geq N$. Using $\dim \mc G = |S|- \dim \ker \phi$, we finally show that
\begin{align}
\log \dim \mc H_{\text{code}}= N-\dim \mc G= N -|S| + \dim \ker \phi = \dim \ker \mc T.
\end{align}

 So when the condition $|S|\geq N$ holds, Theorem~\ref{thm:log} implies $\{p_C\}_{C \in \mc T}$ forms exactly one maximal, mutually-commuting half of $\mc H(S)$ and can be used to define logical qubits by associating $p_C$ with a logical Z operator.  

To summarize how to practically use this result, we consider any topological constraint and some change of boundary. The members of this constraint map non-trivially under the change of boundary, and Theorem \ref{thm:log} assures us their product is in the logical subspace for our original stabilizer code. The definitions of a change of boundary and a topological constraint may seem too technical to be useful, but they are just the rigorous generalization of intuitive ideas. Furthermore, there is no real cost to assuming a constraint is topological so long as one can verify the resulting operator is logical. Theorem \ref{thm:log} is easily demonstrated in the $d=2$ toric code, where we can generate logical string operators via a change of boundary from the torus to the cylinder where the strings are formed on the cylindrical edges. We can also see that the equivalence class of logical operators is a reflection of the equivalence class of topological constraints as well as multiple valid changes of boundary. Moreover, we can go further to use this result. For example, suppose we are given a stabilizer code in three dimensions. If any of the topological constraints are formed by two-dimensional manifolds, then logical operators form one-dimensional strings at the boundary and we can immediately imply a bound on the code distance of $d \lesssim L$, where $L$ is the linear size of the system. This also suggests that any logical operator can be ``deformed'' (via multiplication by some set of stabilizers) to a boundary, even in the case of more exotic logical operators such as those for Haah's code which are typically characterized as fractal operators in the bulk. This result, along with the BrLE rules can be useful for engineering stabilizer codes with certain desirable properties.

\section{Linear Gauge Structures}\label{sec:gauge}

We have demonstrated much of the power and utility of the stabilizer formalism, but the structures presented above are not unique to stabilizer codes. In this section, we abstract the most interesting results to the general case, including continuum theories. Once again, we purposely avoid any discussion of Hamiltonians or Lagrangians and focus on the general mathematical structure. We are also  intentionally avoiding the term `gauge theory' to avoid confusion with Yang-Mills gauge theory which is not encapsulated by what follows. However, we hope to generalize to non-linear gauge structures in the future. To the best of our knowledge, this approach is novel to this paper.

\subsection{Definition of a Linear Gauge Structure}

Consider two vector spaces, $\mc A$ and $\mc F$, both over some field $\mb F$. We refer to $\mc A$ as the {\it potential space} and $\mc F$ as the {\it field space}. Each is equipped with a  linear, non-degenerate form $\Omega: \mc A \to \mc A^*$ and $\Lambda: \mc F \to \mc F^*$ where in general $\mc V^*$ is the dual space of $\mc V$ or more formally 
\begin{align}
\mc V^* = \{ \left(f:\mc V \to \mb F\right) : f \text { is (anti-)linear and bounded}\}.
\end{align}
Non-degeneracy of $\Omega$ requires that for any $A \in \mc A$,  $\Omega (A) =0^*$ if and only if $A=0$, where $0^*$ is the zero map in $\mc A^*$. We require the same holds for $\Lambda$. Non-degeneracy implies that $\Omega$ and $\Lambda$ are  injective. If they are also surjective, then we say they are invertible (as briefly introduced at the end of Section \ref{sec:stab}).

Let $\phi :\mc A \to \mc F$ and $\psi: \mc F \to \mc A$ be some linear maps. We define a linear gauge structure as follows:

\begin{definition}
An $\mb F$-linear gauge structure $GS =\left(( \mc A, \Omega, \phi), (\mc F, \Lambda, \psi)\right)$ satisfies the following:
\begin{enumerate}
\item $\Lambda$ is invertible, and
\item \begin{align}\label{eq:def1}
\phi^\star \Lambda = \Omega \psi,
\end{align}
or the following diagram commutes \footnote{To understand this diagram, we recall that the pullback is an involution for which $\phi^\star: \mc F^* \to \mc A^*$ such that $f \mapsto f\phi$.}:
\begin{equation}
\begin{tikzcd}
\mc F \arrow[r, "\psi"] \arrow[d, "\Lambda"'] & \mc A \arrow[d, "\Omega"] \\
\mc F^* \arrow[r, "\phi^\star"'] & \mc A^*
\end{tikzcd}
\end{equation}
\end{enumerate}
\end{definition} 

We define a linear gauge structure at this level of abstraction in the hopes that we can generalize to non-linear structures. However, we can put \eqref{eq:def1} into the form familiar from Eq. \eqref{eq:stabgauge} by defining $\omega(A,B) =\left( \Omega(B) \right)(A)$ and $\lambda(F, G) = \left(\Lambda(G)\right) (F)$ with which our definition can be written for all $A \in \mc A$ and $F\in \mc F$ as
\begin{align}\label{eq:def2}
\lambda(\phi(A), F) = \omega(A, \psi(F)).
\end{align}
Put this way, we recognize $\psi$ as a generalized adjoint of $\phi$. Also as a consequence of non-degeneracy, $\psi$ is unique to $\phi$, though existence is not guaranteed unless $\Omega$ is also invertible. In all cases considered here, $\Omega$ is invertible and the existence of $\psi$ is guaranteed and given by $\psi= \Omega^{-1} \phi^\star\Lambda$.

\subsection{Interpretation and the Generalized BrLE Rules}

The properties we intend to capture in defining a linear gauge structure were all demonstrated with stabilizers. The primary feature of a gauge theory is that some part of the space of all possibilities has no observable effect and is thus unphysical. A gauge structure analogously encodes this in the map $\phi$, where the kernel is unphysical. Then presumably, $\phi(A)$ is measurable and physical. We interpret this as the physical field. In the same way, $\psi(F)$ ($F$ not necessarily in the image of $\phi$) is another real physical quantity, which we interpret as the charges or currents. We naturally want to relate these measurable quantities to one another and this is provided by the generalized BrLE rules.

\begin{theorem}[Generalized BrLE Rules] \label{th:fund}
For any $\mb F$-linear gauge structure \break $GS =\left(( \mc A, \Omega, \phi), (\mc F, \Lambda, \psi)\right)$,  
\begin{subequations}\label{eq:conss}
\begin{align}
(\ker \phi)^{\perp_\Omega} =& \im \psi, \label{eq:cons}\\
(\im \phi)^{\perp_\Lambda} =& \ker \psi \label{eq:consalt},
\end{align}
\end{subequations}
where $(\ker \phi)^{\perp_\Omega}$ is the subspace $\mc B$ for which $\ker \phi$ is the degenerate subspace of $\Omega \,\iota_{\mc B}$ and likewise for $(\im \phi)^{\perp_\Lambda}$.
\end{theorem}

In simpler terms $(\ker \phi)^{\perp_\Omega}$ is the set of all vectors $B \in \mc A$ such that $\omega(C,B) =0$ for all $C \in \ker \phi$ and $(\im \phi)^{\perp_\Lambda}$ is the set of all vectors $G \in \mc F$ such that $\lambda(F,G) =0$ for all $F \in \im \phi$ (again the more abstract definition in the statement of Theorem \ref{th:fund} is used for future generalization to non-linear structures). The first statement is physically the most important. It states that physical, measurable charges/currents are ``orthogonal'' to non-physical pure gauges. We interpret this as a generalized conservation of charge/current. The second is not so clear, but it is true none-the-less.

\begin{proof}
We prove Eq. \eqref{eq:cons} by showing inclusion one way and then the other.

 ($\supseteq$) Consider any $B \in \im \psi$, which implies there exists an $F \in \mc F$ such that $B= \psi (F)$. Then for any $C \in \ker \phi$, we have $\omega(C, B)= \omega(C, \psi(F)) = \lambda(\phi(C), F) =0$. Thus $B \in (\ker \phi)^{\perp_\Omega}$ and $\im \psi \subseteq (\ker \phi)^{\perp_\Omega}$.

($\subseteq$) Let $\mc B= (\ker \phi)^{\perp_\Omega}$ and $\tilde \phi:\mc F \to \mc A/\ker \phi$ be any linear map such that $\tilde \phi \phi = i_{\mc A/\ker \phi}$ for the quotient map $i_{\mc A/\ker\phi} :\mc A \to \mc A/\ker \phi$ such that $A \mapsto A +\ker \phi$. Note for all $B \in \mc B$ and all $A+\ker \phi \in \mc A/\ker \phi$,
\begin{equation}\label{eq:proof1}
\omega(A+ \ker \phi, B) = \omega(A, B) + \omega(\ker \phi ,B) = \omega(A,B). 
\end{equation}
Thus $ \tilde \Omega : \mc B^* \to \left(\mc A/\ker \phi\right)^*$ is well-defined by $\Omega$ using Eq.~\eqref{eq:proof1}, which implies $i_{\mc A/ \ker \phi}^\star \tilde \Omega = \Omega \iota_{\mc B}$. Now consider the map $\tilde \psi = \Lambda^{-1} \tilde \phi ^\star \tilde\Omega: \mc B \to \mc F$ where we recall that $\Lambda^{-1}:\mc F^* \to \mc F$ exists by the first condition on GS. We then use Eq. \ref{eq:def1} to show that
\begin{align}
\Omega \psi \tilde \psi= \phi^\star \Lambda \tilde \psi =\phi^\star \Lambda \Lambda^{-1} \tilde \phi^\star \tilde\Omega =(\tilde \phi \phi)^\star \tilde \Omega= i_{\mc A/ \ker \phi}^\star \tilde \Omega  = \Omega \iota_{\mc B}.
\end{align}
Let $B \in \mc B$ and $A \in \mc A$ and apply the above equality to show that
\begin{align}
0 =\omega(A, \psi \tilde \psi(B))- \omega(A, B)= \omega(A, \psi \tilde \psi (B) -B) .
\end{align}
As $\Omega$ in non-degenerate, this implies $\psi \tilde \psi (B) =B$ which is to say $\psi \tilde \psi = \iota_{\mc B}$. Therefore, for all $ B \in \mc B = (\ker \phi)^{\perp^\Omega}$ there exists an $F \in \mc F$ such that $\psi(F)= B$, namely $F = \tilde \psi(B)$, which is to say $ (\ker \phi)^{\perp^\Omega} \subseteq \im \psi$. 

We now prove Eq. \eqref{eq:consalt} by showing inclusion one way and then the other.

($\supseteq$) Let $G \in \ker \psi$ and $ F \in \im \phi$, which implies $\psi(G)=0$ and $F=\phi(A)$ for some $A \in \mc A$. Then we have $\lambda(F, G)= \lambda(\phi(A), G) = \omega(A, \psi(G)) =0$. Thus $G \in (\im \phi)^{\perp_\Lambda}$ and we have shown that $\ker \psi \subseteq (\im \phi)^{\perp_\Lambda}$.  

($\subseteq$) Let $G \in  (\im \phi)^{\perp_\Lambda}$. This implies for all $A \in \mc A$, $0 = \lambda(\phi(A), G) = \omega(A, \psi(G))$. As $\Omega$ is non-degenerate, this implies $\psi(G) =0$ which implies $G \in \ker\psi$. Therefore, $(\im \phi)^{\perp_\Lambda} \subseteq \im \psi$.
\end{proof}

This also provides a corollary:

\begin{corollary}
There exist partial right and left inverses, respectively, for $\psi$ and $\phi$, $\tilde \psi:\mc B \to \mc F$ and $\tilde \phi:\mc F \to \mc A/\ker \phi$, which satisfy the relation
\begin{align}
\Lambda \tilde \psi= \tilde \phi^\star \tilde \Omega,
\end{align}
where the restriction is to $\mc B = (\ker \phi)^{\perp_\Omega}$ and $\tilde \Omega$ is as defined in the above proof. They are unique to each other but not unique in general.
\end{corollary} 

This could be useful for stabilizer codes as defining a partial left inverse for $\phi$, $\tilde \phi$, yields a unique partial right inverse for $\psi$, $\tilde \psi$, which one familiar with fault-tolerance quantum computing and memory recognizes as a {\it decoder map} \cite{Terhal2015}. We plan to make use of this in future work. 

\subsection{Example: $U(1)$ Gauge Theory}

The discussion in Section \ref{sec:stabgen} shows that any stabilizer set can be used to form a linear gauge structure and so any one of our stabilizer codes suffices as an example. We can also show that $U(1)$ gauge theory possesses a linear gauge structure. We let $\mc A = L^2(\mb R^4,\mb R^4)$ and $\mc F =L^2(\mb R^4,\mb R^4) \otimes_{\text{AS}}L^2(\mb R^4, \mb R^4)$ where $\otimes_{\text{AS}}$ is the anti-symmetrized tensor product and $L^2$ is the usual Hilbert space of square-integrable functions\footnote{Note that we could just as well have chosen the space of operator-valued or Lie-algebra-valued functions as are used in quantum field theory since all we require are vector space properties. But since we are not using any additional algebraic properties, we stick with classical field theory for simplicity of notation.}. We then let $\omega$ and $\lambda$ be the respective $L^2$ inner-products. $\phi$ is given by (half) the exterior derivative which in component form is $\left(\phi (A)\right)_{\mu \nu} = \hlf( \partial_\mu A_ \nu - \partial_\nu A_\mu)$. From this, we derive $\psi$ uniquely by using Eq. \eqref{eq:def2},
\begin{align}
\lambda(\phi(A), F) =\hlf \int d^4x (\partial_\mu A_\nu - \partial_\nu A_\mu) F^{\mu \nu} = \int d^4x A_\mu \partial_\nu F^{\mu \nu} =\omega(A, \psi(F)).
\end{align}
This implies that in component form, $\psi(F)^\mu =  \partial_\nu F^{\mu \nu}$. Note how this result requires $\mc F$ be anti-symmetrized. We then recognize that the in-homogeneous Maxwell equations are encapsulated in the map $-\psi \phi$. We can also confirm that Eq. \eqref{eq:cons} is correct. A member of  $\ker \phi$ can be written as $(\partial a)_\mu = \partial_\mu a$, for some scalar field $a \in L^2(\mb R^4, \mb R)$. For any member $j \in (\ker \phi)^{\perp_\Omega}$,
\begin{align}
0 = \omega(j, \partial a) = \int d^4x j_\mu \partial^\mu a = - \int d^4x \partial_\mu j^\mu a.
\end{align}
For this to be true for any $a$, it must be that $\partial_\mu j^\mu =0$, which is the continuity equation.

\subsection{Symplectic Gauge Structures}  

It is reasonable to ask why $-\psi \phi$ represents the Maxwell equations, which traditionally is derived from the minimization of a gauge invariant action. Such an action can be expressed in our language as
\begin{align}\label{eq:action}
S(A, J) = - \lambda(\phi(A), \phi(A)) -\omega(A,J) = - \omega( A, \psi \phi(A) +J ),
\end{align}
where we have used linearity. Thus $-\psi \phi(A) = J$ represents the zero of the action (which is not necessarily the minimum). 
%
%
However, this assumes $\lambda(\phi(A), \phi(A)) \neq 0$ which does not hold for $\mb F_2$ (stabilizer) gauge structures. This suggests a more general property a gauge structure might have:

\begin{definition}
An $\mb F$-linear gauge structure $GS =\left(( \mc A, \Omega, \phi), (\mc F, \Lambda, \psi)\right)$ is symplectic if and only if 
\begin{align}
\phi^\star \Lambda \phi = 0_{\mc A \to \mc A^*}.
\end{align}
\end{definition}

where $0_{\mc A \to \mc A^*}$ is the zero map from $\mc A \to \mc A^*$. As before, we can conveniently write this for  all $A,B\in \mc A$ as
\begin{align}
\lambda(\phi(A), \phi(B)) =0.
\end{align}
This is a strong condition to satisfy but when it does, we have the following proof:
\begin{theorem}
An $\mb F$-linear gauge structure $GS =\left( (\mc A, \Omega, \phi), (\mc F, \Lambda, \psi)\right)$ is symplectic if and only if $\im \phi \subseteq \ker \psi$.
\end{theorem}

\begin{proof}
$GS$ is symplectic if and only if for all $A,B \in \mc A$, 
\begin{align}
0=\lambda(\phi(A), \phi(B)) = \omega(A, \psi\phi(B)).
\end{align}
As $\Omega$ is non-degenerate, this is true if and only if $\psi\phi(B)=0$ for all $B$. Therefore, $GS$ is symplectic if and only if $\im \phi \subseteq \ker \psi$.
\end{proof}

A symplectic gauge structure lacks a Maxwell equation, but it does allows us to define the{ \it gauge homology} as
\begin{align}
\mc H(GS)= \ker \psi /\im \phi,
\end{align}
which we recognize as the logical subspace for $\mb F_2$ gauge structures. From this we define a {\it trivial} symplectic gauge structure as one such that the following short sequence is exact:
\begin{equation}
\begin{tikzcd}
0 \arrow[r] & \mc A \arrow[r, "\phi"] & \mc F \arrow[r, "\psi"] & \mc A \arrow[r] &0.
\end{tikzcd}
\end{equation}
A trivial symplectic gauge structure is one such that $\mc H(GS)$ as well as the BrLE rules are trivial.

 In the case of $U(1)$ theory, it is easy to see that
\begin{align}
\lambda_s(F,G)= \int d^4x \epsilon_{\alpha \beta \mu \nu} F^{\alpha \beta} G^{\mu \nu},
\end{align}
makes $\left( (\mc A, \Omega, \phi), (\mc F, \Lambda_s, \psi)\right)$ a symplectic gauge structure and generates the famous topological theta term.\footnote{This is really related to the Chern-Simons differential form, as again, there is no Lagrangian. The use of this form has consequences for $\psi$ as determined by $\phi, \Omega$ and $\Lambda_s$ even if it doesn't change $\ker \phi$.} Furthermore, we can caste this symplectic gauge structure into a recognizable form from $\mb F_2$ gauge structures by noting that $\mc F \simeq  L^2(\mb R^4,\mb R^3)\oplus  L^2(\mb R^4,\mb R^3)$, which is the familiar fact that the electromagnetic tensor is equivalent to the electric and magnetic 3-vector fields. Thus $F^{\mu\nu} \simeq (\vec{F}^1, \vec{F}^2)$ where $F^1_i = F^{0i}$ and $F^{ij}= \epsilon^{ijk} F^2_k$. In this form, $\lambda_s$ is written as
\begin{align}
\lambda_s(F, G)= \int d^4x F^1_i G_i^2 + \int d^4x G^1_i F_i^2 = \omega(F^1, G^2) + \omega(G^1, F^2),
\end{align}
where we have abused our notation for $\omega$. Compared with Eq.\eqref{eq:lamdef2}, we recognize this as the more direct analog of the symplectic form defining commutation for stabilizers. However, it is not an actual symplectic form, but it can be with the modest change of a negative sign between the two terms. This is our general method for constructing symplectic forms: when $\mc F \simeq \mc B \oplus \mc B$ where $\mc B$ has a non-degenerate form $\Omega'$,  and $F \simeq (F^1, F^2)$ and $G\simeq (G^1,G^2)$ for all  $F, G \in \mc F$, we can construct a symplectic form via
\begin{align}
\lambda_s( F, G) = \omega'(F^1, G^2) - \omega'(G^1,F^2). 
\end{align}
It is easy to show that the resulting $\Lambda_s$ is non-degenerate if $\Omega'$ is non-degenerate. Then saying the gauge structure is symplectic is the analogous statement of commutativity of the stabilizer group in the $\mb F_2$ case. 

\subsection{Connection between Gauge Structures and Gauging/Ungauging Maps}

In a paper by Kubica and Yoshida \cite{Kubica2018}, CSS stabilizer codes were described in the language of a gauging and ungauging procedure. In this section, we described the connection between the Kubica and Yoshida description and linear gauge structures. We emphasis that while the two are one-to-one for CSS codes, the symplectic linear gauge structure generalized to non-CSS codes and emphasizes the importance and connection between the $\omega, \lambda, \phi$ and $\psi$ maps.

For an $\mb F_2$-linear gauge structure $\left(( \mc A, \Omega, \phi), (\mc F, \Lambda, \psi)\right)$ which represents a CSS stabilizer code, there is a natural decomposition of the space $\mc A= \mc A_X \oplus \mc A_Z$, where the subspaces are generated by the X-type and Z-type stabilizers, respectively. Furthermore, we have already discussed how $\mc F =\mc F_X \oplus \mc F_Z \simeq \wp(Q)_X \oplus \wp(Q)_Z$. Following Ref.~\cite{Kubica2018} one can define a {\it chain complex} \footnote{As a review, a chain complex is a set of maps $\{\partial_n: \mc V_{n} \to \mc V_{n-1}\}_{n<N}$ called {\it boundary maps} between linear vector spaces $\{\mc V_n\}_{n<N}$ such that $\im \partial_{n+1} \subseteq \ker \partial_{n}$. Similarly, a co-chain complex is a set of maps $\{ \partial^\dagger_n : \mc V_n \to \mc V_{n+1}\}_{n<N}$ called {\it co-boundary maps} such that $\ker \partial^\dagger_{n} \supseteq \im \partial^\dagger_{n-1}$. This allows one to define the $n^{th}$ (co)homology groups $\mc H_n =\ker \partial_n / \im \partial_{n+1}$ ($\mc H^n = \ker \partial^\dagger_n/ \im \partial^\dagger_{n-1}$).},
\begin{equation}
\begin{tikzcd}
\mc A_Z \arrow[r, "\partial_Z"] & \wp(Q)_Z \arrow[r, "\partial_X"]& \mc A_X,
\end{tikzcd}
\end{equation}
where the boundary map $\partial_Z$ maps to the Z support of the stabilizer product and the boundary map $\partial_X$ maps to the error syndromes among the X-type stabilizers. Likewise, we can define the co-chain complex,
\begin{equation}
\begin{tikzcd}
\mc A_X \arrow[r, "\partial_X^\dagger"] & \wp(Q)_X \arrow[r, "\partial^\dagger_Z"]& \mc A_Z,
\end{tikzcd}
\end{equation}
Where the co-boundary maps are defined as the transpose of the boundary maps. From this alone, one can see the relation between this construction and the symplectic gauge structure by identifying
\begin{subequations}
\begin{align}
 \phi=& \partial_X^\dagger \oplus \partial_Z,\\
\psi =& \partial_X \oplus \partial_Z^\dagger,
\end{align}
\end{subequations}
 where by definition of a (co)chain complex $\psi \phi=0$. We can check the braiding condition for $A_X + A_Z \in \mc A_X \oplus A_Z$ and $F_X + F_Z \in \mc F_X \oplus \mc F_Z$,
\begin{align}
\lambda(\phi(A), F)=& \omega(\partial_X^\dagger (A_X), F_Z)+ \omega(\partial_Z(A_Z), F_X)\nonumber \\
 =&  \omega(A_X, \partial_X (F_Z))+ \omega(A_Z, \partial_Z^\dagger(F_X)) =\omega(A, \psi(F)),
\end{align}
where we are using Eq.~\eqref{eq:lamdef2} in the first equality. From this, we see that all other consequences follow. We refer the interested reader to Ref.~\cite{Kubica2018} for more on the gauging and ungauging procedure.

\section{Connecting $\mb F_2$ and $\mb R$ Gauge Structures}\label{sec:connecting}

We now use the machinery from the last two sections to generate a method of constructing one gauge structure from another. We provide a general method for constructing an $\mb R$-linear gauge structure from an $\mb F_2$-linear gauge structure which possess translation symmetry. We use $\mb R$ to describe these theories instead of the more common ``$U(1)$'' terminology to emphasize that unlike $U(1)$ theory, there is no obvious sense in which the continuum gauge structure is related to the $U(1)$ Lie group or any other. In fact, there is no sense of a matter field to be invariant under some group action nor a field theory in general, though one can presumably construct one via Eq.~\eqref{eq:action}. We return to our interpretation of gauge structures in Section \ref{sec:discus}. Thus it is more appropriate to refer to a continuum theory as $\mb R$  (which can easily be generalized to $\mb C$) as this is the field of the gauge structure. 

Due to the results of Section \ref{sec:gauge}, all we need is to construct an $\mb R$ version of $\phi$ which somehow captures the features of the $\mb F_2$ theory and we use symmetry as a guide. As all the stabilizer codes considered here possess translation symmetry, we use the polynomial formalism of Haah, which naturally and compactly captures the translation symmetry of stabilizer codes. Though we review some of the most basic features of this formalism, we refer unfamiliar readers to Ref. \cite{Haah2013} for details.

The polynomial formalism represents both position and translation using monomials of a polynomial ring over $\mb F_2$ which we denote as $\mb F_2[\{e_i\}]$, where $\{e_i\}$ is the set of variables which form the polynomials. The number of variables is the number of cycles or dimensions in the translation group, the power of the monomial represents the position, and similar to before, the product of operators is represented by summation over monomials. For example, suppose we have a 1D cycle of qubits. We can represent all of $\mc P$ with a two dimensional vector of polynomials (one entry for each factor of $\wp(Q)$), and so we represent an example operator using $\mb F_2[e_1]$ as
\begin{align}
z_2 x_3 y_5 \simeq \begin{pmatrix}    
e_1^3 + e_1^5\\
e_1^2 + e_1^5
\end{pmatrix},
\end{align}
where the first entry is for $X$-type operators and the second for $Z$-type operators. We maintain this convention throughout the paper. The usefulness of this formalism is that we can translate such operators by multiplying by a monomial with the appropriate power. Moreover, we can multiply many translated versions of that operator by multiplying the vector by a general polynomial and using distributivity. For example, suppose we want to multiply our example operator by the version of itself translated by two. This is given by
\begin{align}
z_2 x_3 y_5 T^2(z_2 x_3 y_5) \simeq (1 + e_1^2)\begin{pmatrix}
e_1^3 + e_1^5\\
e_1^2 + e_1^5
\end{pmatrix} = 
\begin{pmatrix}
e_1^3 + e_1^7\\
e_1^2 + e_1^4 + e_1^5 + e_1^7
\end{pmatrix}
 \simeq z_2 x_3 z_4 z_5 y_7
\end{align}
Note that the $e_1^5$ term canceled in the first entry as this is sum mod $2$. Since members of the stabilizer group are generated by this exact process of translating and multiplying operators, we can now represent all members of the stabilizer group by these polynomials. We also use the convention that $\bar e_1 = e_1^{-1}$ in what follows. 

In the polynomial formalism, a stabilizer code is entirely encoded in the stabilizer map which is a matrix of some $\mb F_2$ polynomials. Each column represents a stabilizer {\it type} which is defined as the equivalence class of stabilizers related to each other by some translation. For example, there are two types in the $d=2$ toric code, while in the $d=3$ toric code there are four. That is, although all plaquette operators are related by both translation and rotation, we only consider the equivalence class generated by modding out translation, giving us three distinct plaquette types. X-cube also has four and Haah's code has two. We use $m$ for the number of types and so the right dimension of the matrix is $m$. The rows then represent the equivalence classes of Pauli operators related by translation. As each vertex (or general unit cell) has $n$ qubits and there are two independent Pauli operators for each qubit, the left dimension of the matrix is $2n$. We also use the convention that qubits are always on the vertices, so for models originally discussed with one qubit on each edge, we move those qubits in each of the positive coordinate directions to the attached vertex so that $n=d$.  The stabilizer map acts on an $m$-dimensional vector of polynomials which represents some member of $\wp(S)$ and maps to a vector of polynomials representing a Pauli operator in $\mc P$. We recognize the stabilizer map as the polynomial representation of $\phi$. For the $d=2$ toric code, the stabilizer map is a $4\times 2$ matrix of the form 
\begin{align}
\phi_{\text{TC}_2}\simeq \begin{pmatrix}
1+ \bar e_1 & 0\\
1+\bar e_2 & 0\\
0&1 + e_2 \\
0& 1+e_1 \\
\end{pmatrix}.
\end{align}
Similarly, the $d=3$ toric code is a $6 \times 4$ matrix of the form 
\begin{align}
\phi_{\text{TC}_3}\simeq \begin{pmatrix}
1+ \bar e_1 & 0 &0 &0\\
1+\bar e_2 & 0 & 0 &0\\
1+\bar e_3 & 0 & 0 &0\\
0 & 0 & 1+e_3 & 1+ e_2 \\
0& 1+e_3 & 0 & 1+e_1 \\
0& 1+e_2 & 1+e_1 &0
\end{pmatrix}.
\end{align}
To associate this map with an $\mb R$-linear gauge structure, consider that a variable such as $e_1$ in the matrix represents a shift forward in the $\hat e_1$-direction by one, so $1+e_1 =1-e_1$ is a difference between the value at that position and the value at its forward neighbor. This suggests we make the substitution $1+e_1 \to \pm \partial_1$. Going from $\mb F_2$ to $\mb R$ or any larger, finite-dimensional field always involves some sign ambiguity, so we try to make the choice such that the mapping respects the symmetries and maintains the properties of the $\mb F_2$ gauge structure. Using this, we find that our $\mb R$ gauge structure which maps from the $d=3$ toric code is 
\begin{align}
\phi_{\text{TC}_3} \to \begin{pmatrix}
-\partial_1 & 0& 0& 0\\
- \partial_2& 0 & 0 &0\\
- \partial_3 & 0 & 0 &0\\
0 & 0 & -\partial_3 & \partial_2 \\
0& \partial_3 & 0 & -\partial_1 \\
0& -\partial_2 & \partial_1 &0
\end{pmatrix},
\end{align}
which we recognize as the static $U(1)$ $\phi$ map where the top three rows represent the electric field and the bottom three represent the magnetic field. We have the freedom to add time derivatives since stabilizer codes are inherently static. Thus we add the usual $\partial_t \vec A$ term to recover the full $U(1)$ theory. Note that there is no necessity in mapping this way and so there is some ambiguity, primarily in the coefficients. The important thing is to capture the $\mb F_2$ constraint structure with the pure gauge of the $\mb R$-linear gauge structure. For the $d=3$ toric code example, consider any polynomial $a \in \mb F_2[e_1,e_2,e_3]$. Then any trivial constraint can be written in the form 
\begin{align}
\begin{pmatrix}
0\\
1+e_1\\
1+e_2\\
1+e_3
\end{pmatrix} a \to \partial a,
\end{align}
where we have included the time derivative.

We can now turn to the fracton models and apply the same method. The X-cube stabilizer map is once again a $6\times 4$ matrix with the form
\begin{align}\label{eq:contXC}
\phi_{\text{XC}} \simeq \begin{pmatrix}
1+ e_2+e_3 +e_2e_3 & 0& 0& 0\\
1+ e_1 +e_3 + e_1e_3& 0 & 0 &0\\
1+ e_1 + e_2 + e_1e_2 & 0 & 0 &0\\
0 & 0 & 1+\bar e_1 & 1+\bar e_1 \\
0& 1+\bar e_2 & 0 & 1+ \bar e_2 \\
0& 1 + \bar e_3 & 1 + \bar e_3 &0
\end{pmatrix}
\to
\begin{pmatrix}
-\partial_2 \partial_3 & 0& 0& 0\\
-\partial_1 \partial_3&0& 0&0\\
-\partial_1 \partial_2 &0&0 &0\\
0 & 0 & \partial_1 & - \partial_1 \\
0& -\partial_2 & 0 & \partial_2 \\
0& \partial_3 & -\partial_3 &0
\end{pmatrix},
\end{align}
where we have used  $1 + e_1 + e_2 +e_1e_2 = (1+ e_1)(1+e_2)$ and likewise for the other two. So our ``magnetic field'' is $F^1_i = \hlf |\epsilon_{ijk}|\partial_j \partial_k A_0$ and the ``electric field'' is   $F^2_i =\partial_i(A_{i+1} - A_{i+2})$ (conforming to the usual naming convention) where addition in the subscript is mod 3 (starting at 1). As a check, we can show this gauge structure is still symplectic  by considering 
\begin{align}
\lambda_s(\phi(A), \phi(B)) =&\omega' (F^1, G^2) - \omega'(G^1, F^2)= \nonumber \\
=\hlf &\sum_i\left(\int d^3x |\epsilon_{ijk}|\partial_j \partial_k A_0 \partial_i (B_{i+1} -B_{i+2}) - \int d^3x |\epsilon_{ijk}|\partial_j \partial_k B_0 \partial_i (A_{i+1} -A_{i+2}) \right) \nonumber \\
=& -  \int d^3x\partial_x \partial_y \partial_z A_0\sum_i (B_{i+1} -B_{i+2}) + \int d^3x \partial_x \partial_y \partial_z B_0 \sum_i (A_{i+1} -A_{i+2})  \nonumber \\
=&0.
\end{align}
This result should not be surprising given X-cube is a proper stabilizer code.\footnote{Note that the sign which determines whether $\lambda_s$ is symplectic or not did not matter here as was the case with the $d=3$ toric code $\to$ $U(1)$ gauge theory case. This seems to be a consequence of the CSS form of the parent model. If we look back to Eq.\eqref{eq:lamdef2} and consider the commutativity of an $X$-type stabilizer with a $Z$-type stabilizer, one term is trivially zero because the two sets being intersected are both empty, and the other term is zero because the sets overlap an even number of times. So by linearity, the commutativity of general CSS stabilizer group members is always such that the two terms of \eqref{eq:lamdef2} are zero independently.  However, this is not always the case as there exist non-CSS codes. Thus we maintain the negative sign to make $\lambda_s$ a proper symplectic form.} This result is similar to the model suggested in \cite{Ma2018a}.

We now consider the stabilizer map for Haah's code, which is a $4 \times 2$ matrix with the form
\begin{align} \label{eq:contHaah}
\phi_{\text{Haah}} \simeq\begin{pmatrix}
1+e_1 +e_2 +e_3 &0\\
1+ e_1e_2 + e_2e_3 + e_1 e_3 &0\\
0& 1+ \bar e_1 \bar e_2 + \bar e_2 \bar e_3 + \bar e_1 \bar e_3\\
0 & 1 + \bar e_1 + \bar e_2 + \bar e_3
\end{pmatrix} \to 
\begin{pmatrix}
\partial_{[111]}&0\\
\partial^2_{\text{mix}} &0\\
0&\partial^2_{\text{mix}}\\
0 &\partial_{[111]}
\end{pmatrix},
\end{align}
where we have used the definitions from Section \ref{sec:intro}. Again, in this case, it is instructive to check that this is still a symplectic gauge structure: 
\begin{align}
\lambda_s(\phi(A), \phi(B))=& \hlf \int d^3x \left(\partial_{[111]}A_1 \partial^2_{\text{mix}}B_2+\partial_{[111]}B_2 \partial^2_{\text{mix}}A_1\right) \nonumber \\
&-\hlf \int d^3x \left(\partial_{[111]}A_2 \partial^2_{\text{mix}}B_1+ \partial_{[111]}B_1 \partial^2_{\text{mix}}A_2\right)   \nonumber \\
=& \hlf \int d^3x \left(\partial_{[111]}A_1 \partial^2_{\text{mix}}B_2-\partial^2_{\text{mix}}B_2 \partial_{[111]}A_1\right) \nonumber \\
&  -\hlf \int d^3x \left(\partial_{[111]}A_2 \partial^2_{\text{mix}}B_1-\partial^2_{\text{mix}}B_1\partial^2_{\text{mix}}A_2\right) \nonumber \\
=&0,
\end{align}
where we have used integration by parts once on $\partial_{[111]}$ and twice on $\partial^2_{\text{mix}}$, thus giving the three factors of $-1$ necessary for the cancellation to occur. Ref. \cite{Bulmash2018} purposes a similar model for a continuous version of Haah's code which differs by a term that would be equivalent to replacing $\partial^2_{\text{mix}} \to \partial^2_{\text{mix}}  - 2 \partial_{[111]}$. Although consistent with our mapping i.e. $0 \to 2 \partial_{[111]}$, such a term would make the gauge structure no longer symplectic. Recall that the gauge structure being symplectic is a statement about whether the parent model is a proper stabilizer code and whether the gauge homology is well-defined. Thus we do not include such a term.

\subsection{Charge Conservation in $\mb R$ Gauge Structures}

\subsubsection{Toric Code}
As with stabilizer codes, charge conservation is given by the BrLE rules which allows us to understand the properties of the charges without ever considering the form of $\psi$. The $d=2,3$ toric codes are captured by $U(1)$ theory which we have already discussed. However, we can also see the importance of the boundary conditions. For infinite boundary conditions, there are no constant functions in $\mc A$, but if we use periodic boundary conditions, i.e. our spacetime has non-trivial topology, the constant functions are allowed and clearly included in $\ker \phi$. Thus we have the additional BrLE rule for a constant function $a$
\begin{align}
0= \omega(j,a) =\int d^4x a_\mu j^\mu=  a_\mu \int d^4x j^\mu. 
\end{align}
This is only true for all $a$ if the total 4-current is globally conserved. A similar statement of global flux conservation follows if we consider functions in $\ker \phi$ which are dependent on any single direction ( thus constant in orthogonal directions), i.e. $a_0 =0$ and  $\vec{a} = a(\vec v \cdot \vec x)  \vec{v}$ where $\vec v$ is any constant 3-vector and $a$ any function of a single variable. This is in $\ker \phi$ because the resulting electric field is trivially zero and the magnetic field is given by $ a' \epsilon_{ijk} v_jv_v =0$. Again, such a function is only square-integrable in a topologically non-trivial space. We can then apply BrLE rules to find that 
\begin{align}
0 = \omega(j, a) = \int dt \int dv \int du dw \vec j \cdot \vec a= T \int dv  a \left(\int du dw j_\parallel\right),
\end{align}
where $T$ is the period in the time direction, $\{\vec v , \vec u, \vec w\}$ form an orthogonal basis for $\mb R^3$  and $j_\parallel$ is the component of $\vec j$ which is parallel to $\vec v$. So this BrLE rule is true for all such $a$ if and only if $\int du dw j_\parallel =0$ which is the total flux through any plane perpendicular to $\vec v$. 

\subsubsection{X-Cube}

We now turn our attention to the $\mb R$ version of X-cube on a 3-torus. We once again have global charge conservation, but the statements of flux conservation are considerably different and ultimately responsible for fractonic behavior. We start with the $A_0$ field or the magnetic sector. The double derivatives kill any function of a single coordinate direction or linear combinations of them so that the kernel in this sector can be characterized as any function of the form $a_1(x_1) + a_2(x_2) + a_3 (x_3)$. This is analogous to any combination of planar cube constraints in the $\mb F_2$ case. Focusing on one of the coordinate directions, we apply the same logic as with $U(1)$ theory to find that  
\begin{align}
0 = \int d^3x j_0 a_1 = \int dx_1 a_1 \left(\int dx_2 dx_3 j_0\right),
\end{align}
and likewise for the other two directions. This implies $\int dx_i dx_j j_0 =0$ for $i\neq j$. This conservation of charge within intersecting coordinate planes implies that charge must form quadrupole moments at minimum and is a signature of fractonic behavior \cite{Pretko2017a, Ma2018a}.  
For the magnetic sector ($A_1, A_2$ and $A_3$), the kernel can be characterized by vector-valued functions of the form $a_1(x_1) \hat e_1 + a_2(x_2) \hat e_2 + a_3(x_3) \hat e_3$ and $a_0(\vec x) (\hat e_1 + \hat e_2 + \hat e_3)$. The first of these is the analog of the planar topological constraints and the second a reflection of the trivial constraints from the $\mb F_2$ version. 
The BrLE rule for the first case is exactly the same as flux conservation in $U(1)$ theory if we limit $\vec v$ to the coordinate directions. This alone does not yield fractonic behavior. However, we consider the second trivial pure gauge condition where the corresponding BrLE rule requires that $\vec j$ be perpendicular to the $[111]$ direction. For simplicity at a given point, assume $j_1=0$ and $j_2 = -j_3$. In the $\hat e_2, \hat e_3$ plane, this implies that if $\vec j$ points in the second quadrant, then along the $\hat e_1$ direction, there must be an equally long vector in the fourth quadrant (or a sum which is equally long) as depicted in Fig. \ref{contXC}. If we generalize this to a sinusoidal function which fits the periodicity of the system (i.e integrates to zero in the $\hat e_1$ direction), we see that the momentum (wave number) is quantized and confined to one direction. We interpret this as the continuum version which limits an electric particle's motion to the coordinate directions.  

\begin{figure}[t]
\centering
\includegraphics[scale=.4]{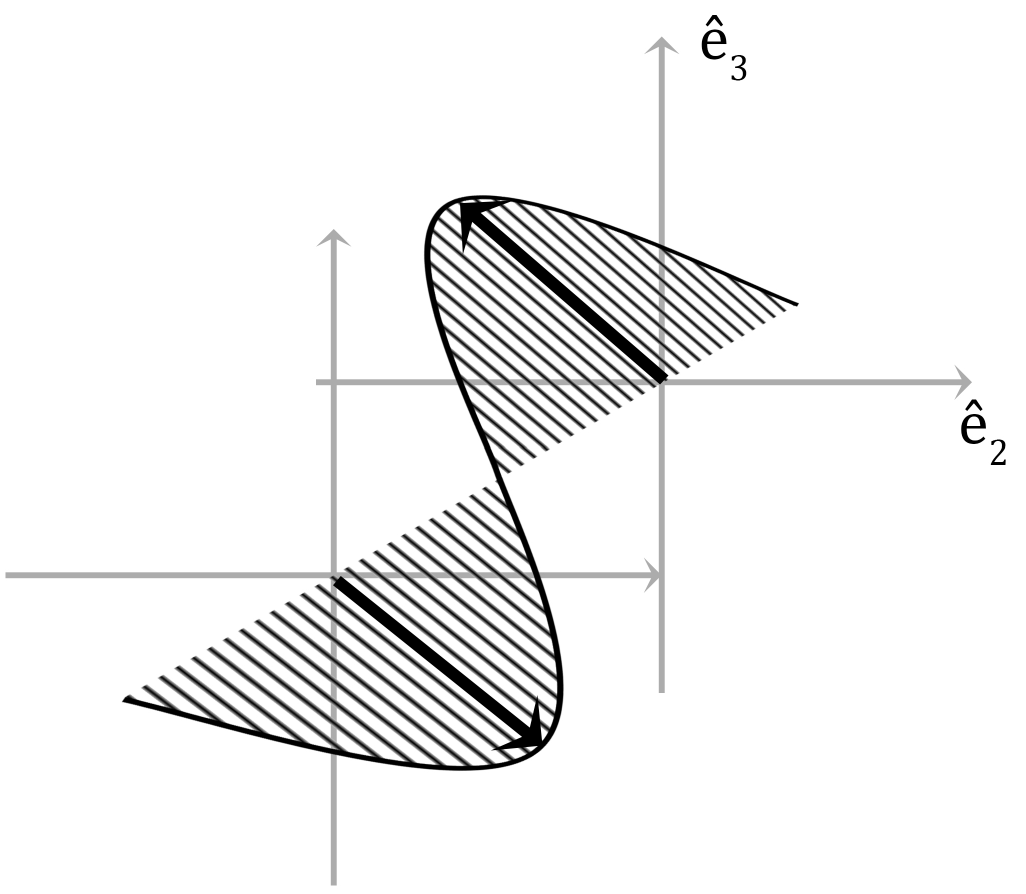}
\caption{Depiction of the fractonic confinement of the electric sector for the $\mb R$ version of the X-Cube model.}\label{contXC}
\end{figure}

\subsubsection{Haah's Code}
\begin{figure}[t]
\centering
\begin{tabular}{cc}
\subfloat[{}\label{recur}]{\includegraphics[scale=1.25]{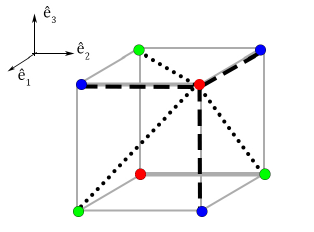}}&
\subfloat[{}\label{haahconst}]{\includegraphics[scale=.5]{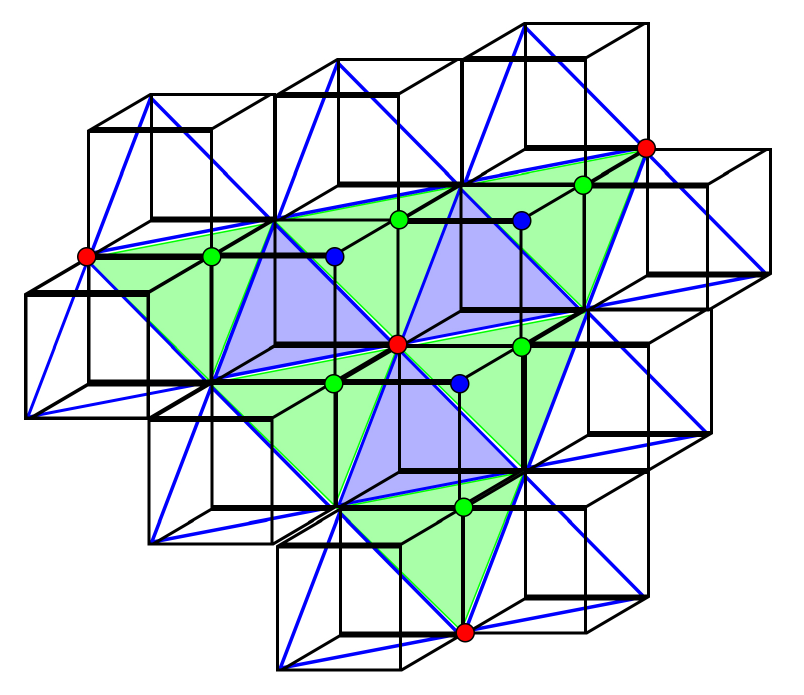}}
\end{tabular}
\caption{Visual depiction of mapping the constraints for Haah's code to coupled layers of the Newman-Moore model. (a) depicts Eqs. \protect \eqref{eq:constraints}; (b) depicts how vertices in the 3D model maps to a plane of the triangular lattice.  }\label{haah1}
\end{figure}

As with the $\mb F_2$ version of Haah's code, understanding $\ker \phi$ is not so obvious.  We focus on the electric sector, i.e. the first field of Eq. \eqref{eq:contHaah}, as the other follows by self-duality. For $a$ to be in the kernel, it must satisfy $\partial_{[111]} a  =\partial^2_{\text{mix}} a =0$.  The first equation implies $a$ is a constant in the $[111]$ direction. As a consequence, the second condition reduces to the 2D Laplace equation in the plane perpendicular to the $[111]$ direction. In theory, this is a complete characterization of $\ker \phi$ so long as we can match the boundary conditions. The first condition alone gives us some characterization of $\im \psi$. If 
\begin{align}
\int d\hat e_{[111]} j_0 =0,
\end{align}
then $j_0 \in \im\psi$ by the BrLE rules. This already suggests the model is not fully fractonic as some excitations are not confined in the $[111]$ direction in the same way the magnetic sector of the X-cube continuum model is not confined in the coordinate directions. But not all members of $\im \psi$ satisfy this condition, so it is possibility a subspace of $\im \psi$ is fully fractonic. 

We can use the continuum model to find some of the constraints for the $\mb F_2$ version. Following \cite{Schmitz2018}, we return to this question from Section \ref{sec:stabgen} by using the polynomial formalism and expand a constraint as $a(e_1,e_2,e_3)= \sum_{\vec{v} \in \mb Z^3} \tilde a(\vec{v}) e_1^{v_1} e_2^{v_2} e_3^{v_3}$. The coefficients must satisfy
\begin{subequations}\label{eq:constraints}
\begin{align}
\tilde a(v_1, v_2, v_3) =& \tilde a(v_1-1, v_2, v_3) +\tilde a(v_1, v_2-1, v_3) +\tilde a(v_1, v_2, v_3-1), \label{eq:cons1}\\
\tilde a(v_1, v_2, v_3) =&\tilde a(v_1-1, v_2-1, v_3) + \tilde a(v_1, v_2-1, v_3-1) +\tilde a(v_1-1, v_2, v_3-1).\label{eq:cons2}
\end{align}
\end{subequations}
These conditions are visually represented in Fig. \ref{recur}. From the continuum solution, we make the ansatz that there exists solutions which are constant in the $[111]$ direction. In the polynomial formalism, this is given by $\tilde a(v_1, v_2, v_3) =\tilde a(v_1+1, v_2+1, v_3+1)$ or vertices connected by the body-diagonal in the $[111]$ direction are equal. We then focus on a single ``Q-bert'' plane of cubes which forms three planes of vertices (technically four, but we only need to consider three of them) as shown in Fig.~\ref{haahconst}. Each of the planes forms a triangular lattice which are stacked and Eqs. \eqref{eq:constraints} couples them. The coupling is reminiscent of the Newman-Moore model \cite{Newman1999}. This model is a classical, planar stabilizer model on the triangular lattice, with stabilizers of the form $\prod_{v \in \triangle}x_v$ where $v$ indexes the vertices forming a triangle in one of the two sub-lattices of triangles (say, the upward pointing triangles). For Haah's code, if the first (red) $[111]$ layer of vertices is mapped to the physical spins of the Newman-Moore model, then the vertices of a layer up (green) are mapped to the parity of the triangle stabilizers by Eq. \eqref{eq:cons1}. The third layer (blue) is also mapped this way to the other triangles of the red layer via Eq. \eqref{eq:cons2}, but also mapped to the second in a similar way by \eqref{eq:cons1}. Finally following from our ansatz, \eqref{eq:cons2} couples the green layer to the red layer. Thus if we can satisfy these simultaneous coupling conditions for these three layers, the rest of the constraint follows by the ansatz. To simplify the picture, we flatten all three layers to a single plane of the triangular lattice, where blue triangles represent the blue layer, green triangles, the green layer and the vertices, the red layer. This is depicted in Fig. \ref{haahconst}. The coupling conditions reduce to the satisfiability of the following rules: the parity of a vertex is equal to the parity of the three triangles of each color which surrounds it and the parity of any triangle must be equal to the parity of the vertices which form it and the parity of the three oppositely colored triangles which surround it. After an exhaustive search, we find maximally 7 independent constraints as given in Fig. \ref{haah2a} (14 when we include the other sector). As a function of the lattice size, the actual number of constraints found for both sectors is 2 for odd lattice sizes, 6 for lattice sizes divisible by 2, and 14 for lattice sizes divisible by 4. The ``star'' pattern represents 2 distinct constraints since for the original cubic lattice, the pattern can be translated by one $[111]$ layer to obtain a distinct pattern. This corresponds to cyclic permutation of the colors or equivalently the three triangular lattices (where we recognize all triangles of a type forms a triangular lattice) for the flattened version. One can also translate/permute one more time in this way, but this is equivalent to the ``sum'' of the other two versions. Fig. \ref{haah2b} show how translated versions of the star pattern in the plane is given by addition with the other patterns.

These constraints are not enough to account for the full number of constraints as it does not scale with the lattice size. Also, if all the constraints did satisfy the ansatz, the excitations would not be confined in the $[111]$ direction, so we should suspect there must be additional constraints. One reason the continuum model does not capture the full fractonic behavior of Haah's code is that the conditions for $\ker \phi$ in the continuum version are symmetric in both the positive and negative $[111]$ direction, whereas it is clear from Fig. \ref{recur} that this is not true for the $\mb F_2$ Haah's code. This is related to the fact that the backward and forward lattice derivatives are not equal. We come to a similar conclusion as Ref. \cite{Bulmash2018} that the lattice is a necessary piece to recover the full fractonic behavior of Haah's code.

Despite not finding all constraints, we can still use what we know to understand why excitations must live at the end of fractal operators. Focusing on the constraint condition Eq. \eqref{eq:cons1}, we can combine this same condition for different points to form a sequence of fractal conditions, i.e. conditions that look like a scaled version of the original. Fig \ref{compcont} shows this process pictorially for the first three generations. A similar sequence can be formed using Eq. \eqref{eq:cons2}. So for any constraint, the sum of the four points in each of these fractal conditions must sum to zero, which implies only an even number of them can be in any constraint. Thus one way for a syndrome to intersect all constraints an even number of times is for the syndrome to contain the four points of any one of these conditions. This is the exact syndrome which is generated by fractal operators as discussed in \cite{Haah2013}. This points out another means of using the BrLE rules. In general, one might object that the BrLE rules are statements which involve the entire system and so they do not apply locally. Indeed, the BrLE rules for Haah's code suggest that for lattice sizes such that the code space is of minimal dimension (it only encodes two logical qubits), either sector can support any syndrome containing only two elements, i.e. it is not globally fractonic. This is true because the only two independent topological constraints are the product of all stabilizers of each type. But any Pauli operator responsible for such a syndrome is almost certainly of the size of the system since such a configuration cannot be generated locally. This is a fair point, but as Fig. \ref{compcont} demonstrates, any constraint must have a manifestation locally. That is, if we consider all stabilizers that are supported in some small, bounded set of qubits, there must exist a product of these stabilizers with no support in this set in order to form a constraint. These types of local rules could be thought of as a kind of local BrLE rules and are related to the Gauss-like laws as discussed in \cite{Schmitz2018}. It would then be these rules which govern how excitations are restricted locally. We hope to make this connection more explicit in future work.       

It is worth noting that the connection of Haah's code to coupled layers of the Newman-Moore model is still viable, where any layer in the stack is coupled to the two layers below it. This could lead to a coupled-layer construction for Haah's code similar to those already found for X-cube \cite{Ma2017, Vijay2017}. We reserve this possibility for future work. 

\begin{figure}[t]
\centering
\begin{tabular}{c}
\subfloat[{}\label{haah2a}]{\includegraphics[scale=.25]{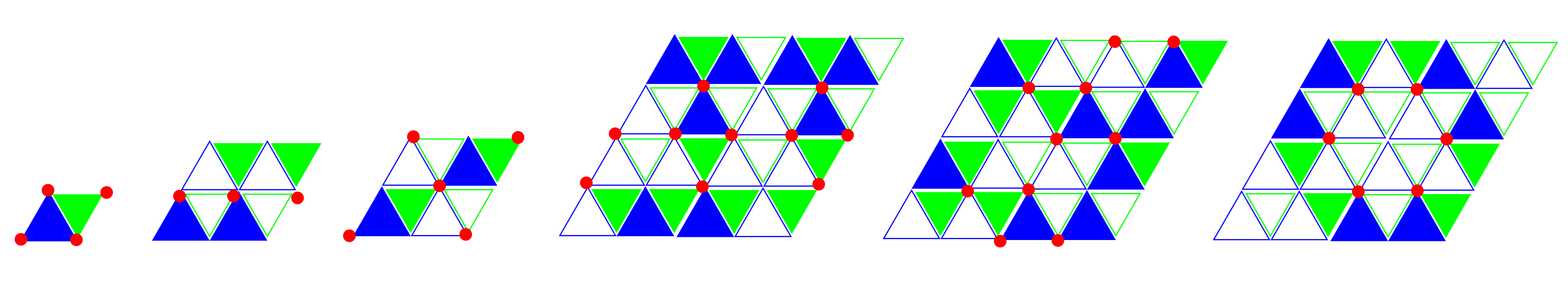}}\\
\subfloat[{}\label{haah2b}]{\includegraphics[scale=.2]{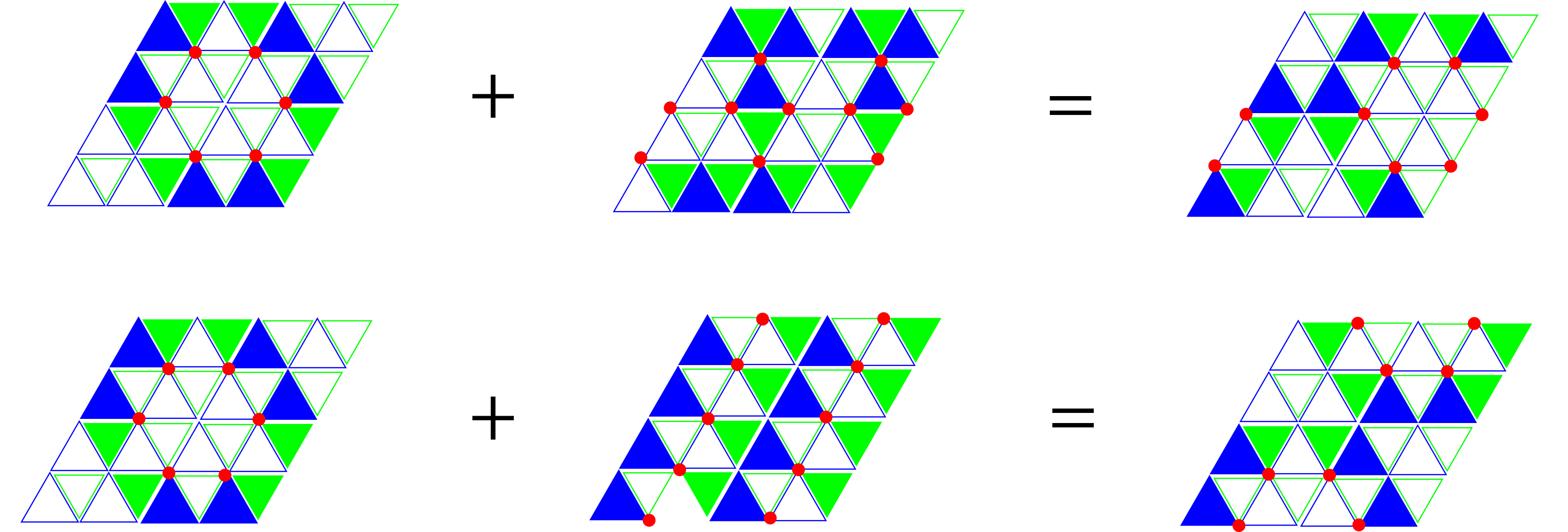}}
\end{tabular}
\caption{ (a) Complete generating set of solutions for the couple Newman-Moore model which maps to the constraints for Haah's code under the ansatz that a solution is constant in the $[111]$ direction. Larger scale solutions are given by tiling these solutions. (b) Two example which demonstrate that translated versions of the patterns are generated by addition with some other pattern.}\label{haah2}
\end{figure}

\begin{figure}[t]
\centering
\includegraphics[scale=.4]{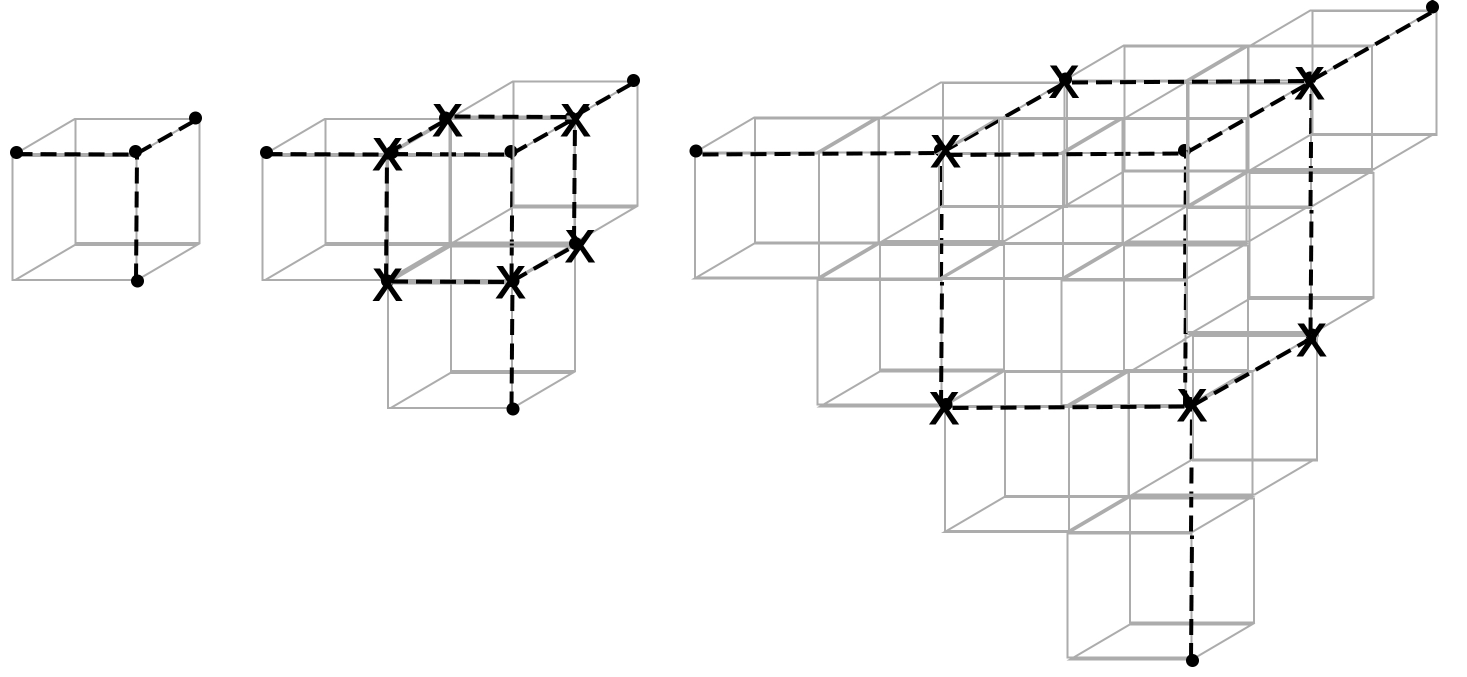}
\caption{Visual depiction of the fractal sequence of constraint conditions using Eq. \protect \eqref{eq:cons1}.}\label{compcont}
\end{figure}

\section{Final Discussion}\label{sec:discus}

In this paper, we have discussed and proven some powerful results for understanding stabilizer codes, their excitations and their logical operators. This led us to abstract the results and introduce the concept of a gauge structure which we then used to connect $\mb F_2$ models to $\mb R$ models. 

Gauge structures seem like a useful means of understanding certain models, especially ones that demonstrate topological properties. Still, we seem to lack an interpretation of the mathematical statements of a linear gauge structure as we can no longer view gauge invariance as a by-product of our demand that matter is invariant under some local group action. However at least in the case of stabilizer codes, we give another interpretation of gauge redundancy as a redundancy in the TPS of the Hilbert space. As stabilizers mutually commute, we can factor our Hilbert space into a TPS such that a given stabilizer is non-trivial in only one of the factors. But the dependency of stabilizers on one another, i.e. the existence of constraints, means this cannot be done consistently for all stabilizers and so there exists a necessary but arbitrary choice in how we factor the space. This is similar to ideas expressed in the context of lattice gauge theory in \cite{Casini2014}. This is also similar to Ref. \cite{Haegeman2015} where more emphasis is put on quantum states as opposed to the Lagrangian or Hamitonian. Excitations are then particles in the sense that they factor in these TPS's, but as a consequence of the factoring redundancy, they must satisfy conservation laws. Note that such conservation laws are much stronger than the typical notion of conservation based upon commutation with the Hamiltonian or enforcement by energetics. In the former cases, conservation is only approximate or spoiled by most perturbations to the Hamiltonian, whereas for the latter case, conservation as given by 
the BrLE rules is a mathematical certainty for states which factorize (exactly or approximately) according to theses TPS's. We might then say the state is in a phase commensurate with the gauge structure. 

Note how this interpretation is complementary to the usual interpretation of $U(1)$ theory. The usual interpretation starts with matter and invents a gauge field to maintain the local redundancy under the group action. In our interpretation, we assume the gauge field, which induces a TPS redundancy, and matter is then understood according to how fields factor with respects to the TPS's. We look to explore this interpretation more deeply in future work. This includes extending the definition of a linear gauge structures to non-linear gauge structures, where the most natural language for this might be category theory.

\section{Acknowledgments}

A special thanks to Rahul M. Nandkishore for encouraging discussion. I would also like to thank  Abhinav Prem and Michael Pretko for enlightening discussion and comments on the manuscript, as well as Sheng-jie Huang, Han Ma, Shriya Pai, Siddharth Parameswaran, and Markus Pflaum. The author is supported by the Air Force Office of Scientific Research under award number FA9550-17-1-0183.

\printbibliography

\end{document}